\documentclass[letter,11pt]{article}

\usepackage[margin=1in]{geometry}
\usepackage{hyperref}
\newcommand{\email}[1]{\href{mailto:#1}{\nolinkurl{#1}}}
\usepackage[numbers]{natbib}
\usepackage{xspace}
\usepackage{xcolor}
\usepackage{graphicx}
\usepackage{amsmath,amssymb,amsthm}
\usepackage{tikz}
\definecolor{dgray}{RGB}{160,160,160}
\definecolor{lgray}{RGB}{235,235,235}
\usetikzlibrary{chains,arrows}
\usepackage{todonotes}
\usepackage{bbm} 
\usepackage{enumerate}
\usepackage[labelfont=bf]{caption}
\usepackage{booktabs} 
\newcommand{\ra}[1]{\renewcommand{\arraystretch}{#1}} 

\newtheorem*{key-definition}{Key Definition (special case)}
\newtheorem*{main-theorem}{Main Theorem (informal)}

\usepackage{color}

\theoremstyle{plain}
\newtheorem{proposition}{Proposition}[section]
\newtheorem{theorem}{Theorem}[section]
\newtheorem{lemma}{Lemma}[section]

\newtheorem{claim}{Claim}[section]
\theoremstyle{definition}
\newtheorem{example}{Example}[section]
\newtheorem{definition}{Definition}[section]
\newtheorem{remark}{Remark}[section]

\newcommand{\Xcomment}[1]{{}}

\DeclareSymbolFont{bbold}{U}{bbold}{m}{n}
\DeclareSymbolFontAlphabet{\mathbbold}{bbold}
\DeclareMathOperator*{\argmax}{arg\,max}
\DeclareMathOperator*{\E}{\mathbb{E}}

\newcommand{\growingmid}{\mathrel{}\middle|\mathrel{}}

\newcommand{\ind}{\mathbbold{1}}
\newcommand{\bx}{\mathbf{x}}
\newcommand{\bv}{\mathbf{v}}
\newcommand{\bb}{\mathbf{b}}
\newcommand{\bp}{\mathbf{p}}
\newcommand{\by}{\mathbf{y}}
\newcommand{\btv}{\mathbf{\tilde{v}}}
\newcommand{\F}{\mathcal{F}}

\newcommand{\ab}{$(\alpha,\beta)$-balanced\xspace}
\newcommand{\abb}{weakly $(\alpha,\beta_1,\beta_2)$-balanced\xspace}
\newcommand{\AB}{$(\alpha,\beta)$-balancedness\xspace}

\newcommand{\ALG}{\mathsf{ALG}\xspace}
\newcommand{\alg}{\mathsf{ALG}\xspace}

\newcommand{\GRE}{\mathsf{GRD}\xspace}
\newcommand{\GRD}{\mathsf{GRD}\xspace}
\newcommand{\OPT}{\mathsf{OPT}\xspace}
\newcommand{\feas}{\mathcal{F}}
\newcommand{\feasx}{\feas_{\bx}}
\newcommand{\feasxell}{\feas_{\bx^\ell}} 

\newcommand{\Ell}{\ensuremath{\mathcal{L}}}

\newcommand{\be}{\begin{equation}}
\newcommand{\ee}{\end{equation}}

\newcommand{\vect}[1]{\ensuremath{\mathbf{#1}}}

\newcommand{\br}{\mathbf{r}}
\newcommand{\bz}{\mathbf{z}}

\def \reals {{\mathbb R}}

\newcommand{\dist}{\mathcal{D}}
\newcommand{\disti}[1][i]{{\mathcal{D}_{#1}}}

\newcommand{\dists}{\vect{\dist}}

\newcommand{\bid}{b}
\newcommand{\bids}{\vect{\bid}}
\newcommand{\bidsmi}[1][i]{\bids_{\text{-}#1}}

\newcommand{\bidiprime}[1][i]{{\bid'_{#1}}}

\newcommand{\val}{v}
\newcommand{\vals}{\vect{\val}}
\newcommand{\valsprime}{\vect{\val}'}
\newcommand{\valsmi}[1][i]{\vals_{\text{-}#1}}
\newcommand{\vali}[1][i]{{\val_{#1}}}
\newcommand{\valiprime}[1][i]{{\val'_{#1}}}

\newcommand{\util}{u}

\newcommand{\utili}[1][i]{\util_{#1}}

\newcommand{\alloc}{x}
\newcommand{\allocs}{\vect{\alloc}}
\newcommand{\allocsprime}{\vect{\alloc}'}

\newcommand{\alloci}[1][i]{{\alloc_{#1}}}
\newcommand{\allociprime}[1][i]{{\alloc'_{#1}}}

\usepackage[ruled,vlined]{algorithm2e}

\begin{document}

\title{Prophet Inequalities Made Easy:\\ Stochastic Optimization by Pricing Non-Stochastic Inputs}

\author{%
	Paul D\"utting%
	\thanks{%
	Department of Mathematics,
	London School of Economics,
	Houghton Street,
	London WC2A 2AE,
	UK,
	Email: \email{p.d.duetting@lse.ac.uk}.} 
	\and
	Michal Feldman%
	\thanks{%
	Blavatnic School of Computer Science,
	Tel Aviv University,
	P.O.B. 39040, Ramat Aviv, Tel Aviv, Israel.
	Email: \email{mfeldman@tau.ac.il}}
	\and
	Thomas Kesselheim%
	\thanks{
	TU Dortmund,
	Otto-Hahn-Str.\ 14,
	44221 Dortmund,
	Germany.
	Email: \email{thomas.kesselheim@cs.tu-dortmund.de}.
	This work was done while the author was at Max Planck Institute for Informatics and Saarland University, supported in part by the DFG through Cluster of Excellence MMCI, and while he was visiting Simons Institute for the Theory of Computing.
	}
	\and
	Brendan Lucier%
	\thanks{
	Microsoft Research,
	1 Memorial Drive \#1, Cambridge, MA 02142, USA.
	Email: \email{brlucier@microsoft.com}}
	}

\maketitle


\begin{abstract}
We present a general framework for stochastic online maximization problems with combinatorial feasibility constraints. The framework establishes prophet inequalities by constructing price-based online approximation algorithms, a natural extension of threshold algorithms for settings beyond binary selection. Our analysis takes the form of an extension theorem: we derive sufficient conditions on prices when all weights are known in advance, then prove that the resulting approximation guarantees extend directly to stochastic settings.  Our framework unifies and simplifies much of the existing literature on prophet inequalities and posted price mechanisms, and is used to derive new and improved results for combinatorial markets (with and without complements), multi-dimensional matroids, and sparse packing problems. Finally, we highlight a surprising connection between the smoothness framework for bounding the price of anarchy of mechanisms and our framework, and show that many smooth mechanisms can be recast as posted price mechanisms with comparable performance guarantees.

\end{abstract}


\section{Introduction}

A concert is being held in a local theatre, and potential audience members begin calling to reserve seats.  The organizer doesn't know individuals' values for seats in advance, but has distributional knowledge about their preferences.  Some need only a single seat, others require a block of seats.  Some think seats are very valuable, others are only willing to attend if tickets are very cheap.  Some prefer front-row seats, some prefer to sit a few rows back, and some prefer the balcony.  The organizer needs to decide which seats, if any, to allocate to each individual as they call.  The goal is to maximize the total value (i.e., social welfare) of the seating arrangement.

%
%
Such stochastic online optimization problems have been studied for decades.
A common goal
is to attain ``prophet inequalities'' that compare the performance of an online algorithm to that of an omniscient offline planner.  A classic result
is that if the goal is to choose exactly one element (i.e., there is only a single seat to allocate), then a simple threshold strategy---choosing the first value higher than a certain pre-computed threshold---yields at least half of the expected maximium value~\cite{KrengelS77,KrengelS78,Samuel84}.
This solution has the appealing property that it corresponds to posting a take-it-or-leave-it price and allocating to the first interested buyer.
A natural question
is whether more complex allocation problems (like the concert example above) can be approximated by posting prices and allowing buyers to select their preferred outcomes in sequence.

%
%
Driven in part by this connection to posted prices, prophet inequalities have seen a
resurgence in theoretical computer science. Recent work has established new prophet inequalities
for a variety of allocation problems, including matroids \cite{ChawlaHMS10,KleinbergW12}, unit-demand bidders \cite{ChawlaHMS10,Alaei14}, and combinatorial auctions \cite{FeldmanGL15}.
In this paper we develop a 
framework for proving prophet inequalities and constructing posted-price mechanisms.
Our framework, which is based on insights from economic theory, unifies and simplifies many existing results and gives rise to new and improved prophet inequalities in a host of online settings.

\subsection{Example: Combinatorial Auctions}
\label{sec:intro.example}

To introduce our framework we will consider a combinatorial auction problem.  There is a set $M$ of $m$ items for sale and $n$ buyers. 
Each buyer $i$ has a valuation function $v_i\colon 2^M \to \mathbb{R}_{\geq 0}$ that assigns non-negative value to every subset of at most $d$ items.\footnote{Alternatively, we can suppose that there is a cardinality constraint that no buyer can receive more than $d$ items.}  Valuations are non-decreasing and normalized so that $v_i(\emptyset) = 0$, but otherwise arbitrary.  The goal is to assign items to buyers to maximize total value.
Write $\bv(\bx) = \sum_{i=1}^n v_i(x_i)$ for the total value of allocation $\bx = (x_1, \dotsc, x_n)$, where $x_i \subseteq M$ for all $i$.  There is a simple $O(d)$-approximate  greedy algorithm for this problem and a lower bound of $\Omega(d / \log d)$ assuming $P \neq NP$ \cite{Trevisan01}.  Our goal is to match this $O(d)$ approximation as a prophet inequality with posted item prices.
That is, given distributions over the valuations, compute prices for the items so that, when buyers arrive in an arbitrary order and each chooses his most-desired bundle from among the unsold items, the expected total value is an $O(d)$ approximation to the expected optimum.\footnote{There is a straightforward lower bound of $\Omega(d)$ on the approximation of any posted item prices.  Suppose there are $d$ items and two agents.  The first agent is unit-demand and has value $1$ for any single item.  The second agent values the set of all $d$ items for value $d$, and has value $0$ for any subset.  If all items have price greater than $1$, then neither agent purchases anything.  If any item has price less than $1$, then the unit-demand agent (who chooses first) will purchase the cheapest single item and the other agent will purchase nothing, generating a total value of $1$ whereas the optimum is $d$.  One can avoid issues of tie-breaking by perturbing the values by an arbitrarily small amount.}

Let's first consider the simpler full information case where all valuations are known in advance.
This problem is still non-trivial, and in fact
there may not exist prices that lead to the optimal allocation.\footnote{For example, suppose there are three items and four single-minded bidders.  The first three bidders each have value $2$ for a different pair of items, and the last bidder has value $3$ for the set of all three items, so at most one bidder can get positive value, and it is optimal to allocate all items to the last bidder.  However, at any item prices where the last bidder is willing to purchase, one of the other bidders will purchase first if arriving before the last bidder. This leads to a $3/2$ approximation in the worst arrival order.}
%
Intuitively, what we need for an approximation result are prices that balance between two forces.  They should be small enough that high-valued buyers are willing to purchase their optimal bundles if available, but also large enough that those items will not first be scooped up by bidders with much lower values.
Such ``balanced'' prices can be obtained as follows: Given valuation profile $\bv$, consider the welfare-maximizing allocation $\bx^*$ (which we can assume allocates all items). Then for each item $j$, say $j \in x^*_i$, set the price of $j$ to $p_j = v_i(x^*_i)/2|x^*_i|$.
These prices are low enough that the total price of all items is at most $1/2 \cdot \bv(\bx^*)$, which is significantly less than the total value of $\bx^*$. At the same time, prices are high enough that, for any set of goods $S$, the total price of $S$ is at least $1/2d$ of the value of allocations in the optimal allocation $\bx^*$ that intersect $S$.  So, in particular, a bidder that purchases $S$ must have value at least that high.

To see why these prices yield an $O(d)$ approximation, let $\bx$ denote the purchase decisions of the players and let $I \subseteq N$ be the set of players $i$ such that $x^*_i$ intersects with $\bx$. The welfare achieved by $\bx$ is equal to the revenue generated plus the sum of buyer utilities.  The revenue is the sum of prices of the items sold, and since prices are ``balanced'' this is at least $(1/2d) \cdot \sum_{i \in I} v_i(x^*_i)$.  Also, each buyer $i \not\in I$ could have chosen to purchase $x_i^*$, and therefore must get at least as much utility as they would by purchasing $x^*_i$, which is $v_i(x^*_i)$ minus the price of $x^*_i$.  Again, since prices are balanced, this means the sum of buyer utilities is at least $\sum_{i \not\in I} v_i(x^*_i) - 1/2 \cdot \bv(\bx^*)$.
Multiplying this by $1/2d$ and adding the revenue gives an $O(d)$ approximation.\footnote{
For simplicity we assumed here that $\sum_{i \not\in I} v_i(x^*_i) - 1/2 \cdot \bv(\bx^*) \geq 0$.
More generally, since utilities are non-negative, the sum of buyer utilities is at least $\max\{\sum_{i \not\in I} v_i(x^*_i) - 1/2 \cdot \bv(\bx^*), 0\}$. If the maximum is attained at $0$, then $\sum_{i \in I} v_i(x^*_i) > 1/2 \cdot \bv(\bx^*)$ and the revenue alone exceeds $(1/4d) \cdot \bv(\bx^*)$, as desired.}.

The argument above was for the full information 
case.  Perhaps surprisingly, the existence of sufficiently ``balanced'' prices for full information
instances \emph{also} establishes an $O(d)$-approximate prophet inequality for the general stochastic problem, where one has only distributional knowledge about valuations.
Our main result is this reduction from the stochastic setting to the full information 
setting, which holds for a broad class of allocation problems.

\subsection{A Framework for Prophet Inequalities}


Consider a more general combinatorial allocation problem, where the cardinality constraint $d$ is replaced with an arbitrary downward-closed feasibility constraint $\feas$ and each $v_i$ is drawn independently from an arbitrary distribution $\mathcal{D}_i$.
While our framework applies for more general outcome spaces (see Sections~\ref{sec:model} and~\ref{sec:extension}), combinatorial allocation problems provide a sweet spot between expressiveness and clarity.
Our key definition is the following notion of balanced prices for full-information instances.
For each $\bx \in \feas$ we write $\OPT(\bv \mid \bx)$ for the \emph{optimal residual allocation}: the allocation that maximizes $\sum_i v_i(x'_i)$ over $\bx' \in \feas$ with $\bx, \bx'$ disjoint and $\bx \cup \bx' \in \feas$.
Given a fixed valuation profile $\bv$, a \emph{pricing rule} defines a price $p^{\bv}_i(x_i)$ for every bundle that we can assign to buyer $i$. For example, the item prices described in Section~\ref{sec:intro.example} define a pricing rule $p^{\bv}_i(x_i) = \sum_{j \in x_i} p_j$.  Below we also extend the definition to dynamic prices, i.e., prices that depend on which allocations have already been made.

\begin{key-definition}[($\alpha,\beta$)-balanced prices]
Let $\alpha, \beta > 0$.
A pricing rule $\bp^{\bv} = (p_1^\bv,\dots,p_n^\bv)$ defined by functions $p_i^\bv: 2^M \rightarrow \mathbb{R}_{\ge 0}$ is \emph{$(\alpha, \beta)$-balanced} with respect to valuation profile $\bv$ if for all $\bx \in \feas$ and all $\bx' \in \feas$ with $\bx,\bx'$ disjoint and $\bx \cup \bx' \in \feas$,
\begin{enumerate}
\item[(a)] $\sum_i p_i^{\bv}(x_i) \geq \frac{1}{\alpha} ( \bv(\OPT(\bv)) - \bv(\OPT(\bv \mid \bx)))$ ,
\item[(b)] $\sum_i p_i^{\bv}(x_i') \leq \beta \, \bv(\OPT(\mathbf{v} \mid \mathbf{x}))$ .
\end{enumerate}
\end{key-definition}

The first condition formalizes what it means that prices are high enough: the sum of prices for $\bx$ should partially cover the welfare lost due to allocating $\bx$.  The second condition formalizes ``low enough:'' the sum of prices for any $\bx'$ that is still feasible ``after'' allocating $\bx$ should not be much higher than the optimal residual welfare.

Our 
main result is that 
the existence of balanced prices for 
full information instances directly implies 
a price-based prophet inequality for the stochastic setting.
The idea to choose balanced prices is a natural one and has appeared in the prophet inequality literature before, most explicitly in the notion of balanced thresholds of Kleinberg and Weinberg \cite{KleinbergW12}. 
Previous definitions, however, applied to the stochastic setting directly, which made the construction and analysis of balanced thresholds inherently probabilistic.  A main advantage of our framework is that it suffices to reason about the simpler full-information setting.


\begin{main-theorem}
Consider the setting where valuations are drawn from product distribution $\mathcal{D}$. Suppose that the pricing rule $\bp^{\bv}$ is $(\alpha, \beta)$-balanced with respect to valuation profile $\bv$. Then posting prices
\[
p_i(x_i) = \frac{\alpha}{1 + \alpha \beta} \mathbf{E}_{\tilde{\bv} \sim \mathcal{D}} \left[ p_i^{\tilde{\bv}}(x_i) \right]
\]
achieves welfare at least $\frac{1}{1 + \alpha \beta} \mathbf{E}[\bv\left(\OPT(\bv)\right)]$. 
\end{main-theorem}


In other words, to construct appropriate prices for a stochastic problem instance, it suffices to construct balanced prices for the full-information instances in its support and then post the expected values of those prices, scaled by an appropriate factor.
The proof of our main theorem is similar in spirit to proofs in the Price of Anarchy literature \cite{Roughgarden15,SyrgkanisT13} or for establishing algorithmic stability \cite{HardtRS16}, in that it uses ``ghost samples.'' It is, however, considerably more involved because of the sequential, online aspect of our problem. 

\begin{remark}[Weakly Balanced Prices]
We also define a notion of weakly balanced prices, in which it suffices to upper bound the prices by $\beta \bv(\OPT(\bv))$. In this case, we can show that posting an appropriately scaled version of the expected prices yields a $1/4\alpha\beta$-approximate prophet inequality.
\end{remark}

\begin{remark}[Computation]
It is sometimes easier to compute prices that are balanced with respect to an approximation algorithm $\ALG$ rather than $\OPT$.  Our result still applies in this case, with $\OPT$ replaced by $\ALG$ in the welfare guarantee.
We also note that if the price rule $\bp$ in the main theorem is perturbed to some $\hat{\bp}$ with $||\bp - \hat{\bp}||_{\infty} < \epsilon$, then the welfare guarantee degrades by at most an additive $O(n\epsilon)$ term.  This robustness is desirable in itself, and also implies that appropriate prices can be computed for bounded values with $\textsc{poly}(n,m,1/\epsilon)$ samples using standard concentration bounds, as has been observed for various posted price settings~\cite{ChawlaHK07,FeldmanGL15}.
\end{remark}

\begin{remark}[Static vs. Dynamic, Anonymous vs. Discriminatory, Bundle vs. Item Pricing]
We have described our framework for static, discriminatory, bundle prices. In general, our construction has the property that if the full-information balanced prices $\bp^{\bv}$ 
are dynamic, anonymous, and/or take the form of item prices, then the derived prices for the stochastic setting will have these properties as well.
For example, our result holds also for dynamic prices, replacing $p_i(x_i)$ and $p_i(x'_i)$ with $p_i(x_i \mid \bx_{[i-1]})$ and $p_i(x'_i \mid \bx_{[i-1]})$ where the conditioning on $\bx_{[i-1]}$ indicates that the price to player $i$ may depend on the purchase decisions of players that precede him. See Sections \ref{sec:model} and \ref{sec:extension} for details.
\end{remark}

\begin{remark}[Arrival Order]
Balancedness can depend on player arrival order.  In the applications we consider, our results hold even if the arrival order is chosen by an adaptive adversary that observes previous realized values and purchase decisions before selecting the next player to arrive. 
\end{remark}

Let's return to our example from Section~\ref{sec:intro.example}.
We established the \emph{existence} of weakly $(d,1)$-balanced prices (simply undo the scaling by $1/2$), so our main result implies a $O(d)$-approximate prophet inequality.
What about computation?  We can compute prices in polynomial time by basing them on the $O(d)$-approximate greedy algorithm (rather than the optimal allocation), but then we only get a $O(d^2)$-approximate solution.  It turns out that we can further improve this to $O(d)$ in polynomial time, as we hoped for in Section~\ref{sec:intro.example}, by applying our main theorem to a fractional relaxation of the auction problem.  See Section~\ref{section:applications} for more details.

\vspace*{-10pt}

\paragraph{Composition} We also show that balanced prices ``compose'', as was shown for mechanism smoothness in \cite{SyrgkanisT13}. This means that to derive a prophet inequality for a complex setting it often suffices to show balancedness for a simpler problem.
See Appendix \ref{app:composition}.

\subsection{Unification of Existing Prophet Inequality Proofs}

Our framework unifies and simplifies many of the existing prophet inequality proofs.
We list some representative examples below.
We discuss the first example in more detail in Appendix \ref{app:single-item}.
The other two examples are covered
in Appendices \ref{app:cas} and \ref{app:matroid}.

\begin{example}[Classic Prophet Inequality, \cite{KrengelS77,KrengelS78}]
\label{ex:single-item}
The goal is to pick the single highest-value element $v_i$. The pricing rule $\bp^\bv$ defined by $p_i^\bv(x_i) = \max_i v_i$ for all $i$ is $(1,1)$-balanced.
\end{example}

\begin{example}[Matroids, \cite{KleinbergW12}]
\label{ex:matroid}
The goal is to pick a maximum weight independent set in a matroid. Encode sets $S$ by $n$-dimensional vectors $\bx$ over $\{0,1\}$ such that $x_i = 1$ if $i \in S$. Then one can define a dynamic pricing rule $\bp^\bv$ by $p_i(x_i \mid \by) = \bv(\OPT(\bv \mid \by)) - \bv(\OPT(\bv \mid \by \cup x_i))$ for all $i$, where $\by$ is the set of previously-selected elements. This pricing rule is $(1,1)$-balanced.
\end{example}

\begin{example}[XOS Combinatorial Auctions, \cite{FeldmanGL15}]
\label{ex:xos}
The goal is to assign $m$ goods to $n$ buyers with XOS valuations\footnote{A valuation $v$ is XOS if there is a collection of additive functions $a_1(\cdot),\ldots ,a_k(\cdot)$, such that for every set $S$, $v(S) = \max_{1 \leq i \leq k} a_i(S)$.  This is a generalization of submodular valuations \cite{LehmanLN01}.}. Let $\bx^* = \OPT(\bv)$ and let $a_1, \dots, a_n$ be the corresponding additive supporting functions. Set item prices $p_j = a_i(j)$ for $j \in x^*_i$.
This pricing rule is $(1,1)$-balanced.
\end{example}

The final example illustrates the power of our composition results (see Appendix \ref{app:composition}): the existence of $(1, 1)$-balanced prices for XOS combinatorial auctions, and hence a $2$-approximate prophet inequality, follows directly from the existence of $(1, 1)$-balanced prices for a single item, despite being significantly more complex.  It also yields a $O(\log m)$-approximate prophet inequality for subadditive valuations by approximating subadditive valuations with XOS valuations \cite{Dobzinski07,BhawalkarR11}.


\subsection{New and Improved Prophet Inequalities}
\label{sec:intro.new}

We also establish new prophet inequalities using our framework; see Table~\ref{table:applications}.
%
Our first result is a poly-time $(4k-2)$-approximate prophet inequality for MPH-$k$ combinatorial auctions\footnote{The maximum over positive hypergraphs-$k$ (MPH-$k$) hierarchy of valuations \cite{FeigeFIILS15} is an inclusive hierarchy, where $k$ measures the degree of complementarity.}.

\begin{table*}\centering
\ra{1.3}
\begin{tabular}{@{}p{2.4cm}p{1.9cm}p{5cm}p{3.5cm}p{1.6cm}@{}}
\toprule
\textbf{Feasibility Constraint} & \textbf{Valuation Class} & \textbf{Pricing Model} & \textbf{Upper Bound} & \textbf{Query Model} \\
\midrule
Combinatorial Auction & XOS & Static, anonymous item prices & \mbox{$\frac{2e}{e-1}$~\cite{FeldmanGL15}} \hspace{1.1cm} 2~[this work] & XOS, \ Demand  \\
Combinatorial Auction & MPH-$k$ & Static, anonymous item prices & \mbox{$O(k^2)$~\cite{FeldmanGL15}} \hspace{1.5cm} $4k-2$~[this work] & MPH, \hspace{1cm}Demand  \\
Matroid & Submodular & Dynamic prices & $2\text{ (existential)}$  \hspace{1cm} $4$ (computational) & Value  \\
Knapsack & Additive & Static, anonymous prices & $3$ & Explicit \\
$d$-Sparse PIPs & Additive & Static, anonymous prices & $8d$ & Explicit \\
\bottomrule
\end{tabular}
\caption{Overview of applications. Results are computational unless otherwise stated.  The query model refers to the valuation access needed for the computational upper bounds, where ``explicit'' indicates that valuations can be described explicitly. All results are order oblivious (see Section \ref{sec:model}).}
\label{table:applications}
\vspace*{-10pt}
\end{table*}

\begin{theorem}[Combinatorial auctions with MPH-$k$ valuations]
\label{thm:mphk}
For combinatorial auctions with MPH-$k$ valuations, a $(4k-2 + \epsilon)$-approximate posted-price mechanism, with static item prices, can be computed in $\textsc{poly}(n,m,1/\epsilon)$ demand and MPH-$k$ queries.
\end{theorem}

Theorem \ref{thm:mphk} improves the poly-time result of \cite{FeldmanGL15} from $O(k^2)$ to $O(k)$.
We note two interesting special cases. 
First, combinatorial auctions with bundle size $d$ (from Section~\ref{sec:intro.example}) belong to MPH-$d$, so Theorem~\ref{thm:mphk} captures the polytime $O(d)$ approximation discussed above.
Second, XOS valuations coincide with MPH-$1$,  
so Theorem \ref{thm:mphk} improves the previously best known poly-time result of \cite{FeldmanGL15} from
$2e/(e-1)$ to $2$, matching the existential lower bound.
See Section \ref{section:applications} and Appendix \ref{app:cas}.



The second set of new results includes Knapsack feasibility constraints and $d$-sparse Packing Integer Programs (PIPs), for which we obtain a constant- and a $O(d)$-approximation, respectively. These settings are presented in Sections \ref{section:applications} and \ref{app:sparse-pip}, respectively. 
%

\begin{theorem}[Knapsack]
\label{thm:knapsack}
For Knapsack constraints, a factor $(5+\epsilon)$-approximate posted-price mechanism, with static prices, can be computed in $\textsc{poly}(n,1/\epsilon)$.
This improves to a $(3+\epsilon)$ approximation if no individual demands more than half of the total capacity.
\end{theorem}

\begin{theorem}[Sparse PIPs]
\label{thm:sparse-pip}
For $d$-sparse Packing Integer Programs (PIPs) with constraint matrix $\mathbf{A} \in \mathbb{R}^{m \times n}_{\ge 0}$ where $a_{j,i} \leq 1/2$ for all $i,j$ and unit capacities, a factor $(8d+\epsilon)$-approximate posted-price mechanism, with static prices, can be computed in time $\textsc{poly}(n,m,1/\epsilon)$. 
\end{theorem}

To the best of our knowledge, Theorems~\ref{thm:knapsack} and~\ref{thm:sparse-pip} are the first prophet inequalities for these settings.  We note that~\cite{FeldmanSZ15} derived a prophet inequality for closely-related fractional knapsack constraints, with approximation factor $\approx 11.657$.  We obtain an improved prophet inequality for this fractional setting: a corollary of Theorem~\ref{thm:mphk} (with $k=1$) is that one can obtain a $2$-approximation for a fractional knapsack constraint using a static per-unit price, even when knapsack weights are private and arbitrarily correlated with buyer values.  See Section~\ref{section:applications} for more details.


Finally, in Appendix \ref{app:matroid} we generalize the matroid prophet inequalities of \citet{KleinbergW12} to settings where players make choices regarding multiple elements of a matroid, and have submodular preferences over subsets of elements.

\begin{theorem}[Multi-Dimensional Matroids]
\label{thm:submodular-matroid}
For matroid feasibility constraints and submodular valuations, there is a $(4+\epsilon)$-approximate posted-price mechanism, with dynamic prices, that can be computed in $\textsc{poly}(n,1/\epsilon)$ value queries.
\end{theorem}


%
%
%


\subsection{From Price of Anarchy to Prophet Inequalities}

In the proof sketch in Section~\ref{sec:intro.example}, we derived a lower bound on buyer utility by considering a deviation to a certain purchasing decision.  This deviation argument, which appears in the proof of our main result, is also useful for establishing Price of Anarchy bounds~\cite{Roughgarden15,SyrgkanisT13}.
%
%
%
There is a subtle but important difference, however.  In smoothness proofs one considers deviations against a fixed strategy profile, while the prophet inequality problem is inherently temporal and agents deviate at different points in time.
As it turns out, many smoothness proofs have a built-in charging scheme (which we refer to as outcome smoothness) that, under the assumption that critical payments are monotonically increasing, implies prophet inequalities with the same (asymptotic) approximation guarantee.  We provide examples showing that both outcome smoothness and monotonicity are necessary for this result to hold.  We also provide two ``black-box reductions'' for binary single-parameter settings, where PoA guarantees of $O(\gamma)$ established by (normal) smoothness imply $O(\gamma^2)$-approximate prophet inequalities.  See Section~\ref{sec:smoothness}.

\begin{theorem}[informal]
For general multi-parameter problems, if the first-price (i.e., pay-your-bid) mechanism based on declared welfare maximization has a Price of Anarchy of $O(\gamma)$ provable via outcome smoothness, and critical payments are monotonically increasing, then posting a scaled version of the critical payments yields a $O(\gamma)$-approximate price-based prophet inequality.
\end{theorem}


Using these results we can, for example, rederive the classic prophet inequality \cite{KrengelS77,KrengelS78} from the smoothness of the first-price single-item auction~\cite{SyrgkanisT13} or the matroid prophet inequality \cite{KleinbergW12} from the smoothness of the pay-your-bid, declared welfare maximizing mechanism for selecting a maximum-weight basis~\cite{LucierS15}.

\subsection{Further Related Work}


Prophet inequalities and their applicability as posted-price mechanisms were (re-)discovered in theoretical computer science by \cite{HajiaghayiKS07}. Subsequently, threshold-based prophet inequalities and posted-price mechanisms were developed for matroids and matroid intersection  \cite{ChawlaHMS10,KleinbergW12,AzarKW14}, polymatroids \cite{DuettingK15b}, unit-demand bidders \cite{ChawlaHMS10,Alaei14}, and combinatorial auctions \cite{Alaei14,FeldmanGL15}.

Not all prophet inequalities in the literature are based on explicit thresholds.
Examples include prophet inequalities for the generalized assignment problem \cite{AlaeiHL12,AlaeiHL13}, matroids and matroid intersection \cite{FeldmanSZ16}, and for general binary feasibility constraints \cite{Rubinstein16}.
%
On the other hand, many posted-price mechanisms from the literature are constructed either without explicit reference to prophet inequalities or via different techniques.  Chawla et al.~\cite{ChawlaHK07} developed approximately-optimal (revenue-wise) posted-price mechanisms for unit-demand buyers.  Posted-price mechanisms have subsequently been developed for a variety of other auction settings~\cite{DobzinskiNS06,BalcanBM08,ChakrabortyHK09,AlaeiHNPY15}.
Dynamic posted prices that give optimal welfare for unit-demand buyers were established in \cite{CAEFF16}.
Recently, dynamic posted prices for various online settings have been considered, including $k$-server on the line and metrical task systems \cite{CohenEFJ15}, and makespan minimization for scheduling problems \cite{FeldmanFR17}.


Most recently, and in parallel to this work combinatorial prophet inequalities were developed in \cite{RubinsteinS17} and \cite{CaiZ17}. The former, amongst others, proves prophet inequalities for subadditive CAs, but considers a different allocation model
and is therefore imcomparable. The latter, in turn, focuses on revenue and not welfare as we do here. Finally, \cite{AbolhassaniEEHKL17} re-considers the classic prophet inequality setting, but in a large market setting and assuming random or best order.

The notion of smooth games was introduced by Roughgarden \cite{Roughgarden15} as a tool for bounding the price of anarchy, which measures the inefficiency that can be incurred in equilibrium. This notion has been extended to mechanisms by Syrgkanis and Tardos \cite{SyrgkanisT13}.
Notions of outcome smoothness were considered in \cite{DaskalakisS16,LykourisST16}.

\section{General Model and Notation}
\label{sec:model}

\paragraph{Problem Formulation} 
There is a set $N$ of $n$ agents.  
For each agent $i \in N$ there is an outcome space $X_i$
containing 
a null outcome $\emptyset$.  We write $X = X_1 \times \dotsc \times X_n$ for the joint outcome space.
Given outcome profile $\bx \in X$ and a subset of agents $S \subseteq N$, we will write $\bx_S$ for the outcome in which each $i \in S$ receives $x_i$ and each $i \not\in S$ receives $\emptyset$.  
Specifically, we will write $\bx_{[i-1]}$ for allocation $\bx$ with the outcomes of agents $i, \dotsc, n$ set to $\emptyset$.
There is a subset $\feas \subseteq X$ of \emph{feasible} outcomes.  We will assume that $\feas$ is downward-closed, so that if $\bx \in \feas$ then also $\bx_S \in \feas$ for all $S \subseteq N$. 


A \emph{valuation function} for agent $i$ is a function $\vali \colon X_i \to \reals_{\geq 0}$.  We will assume values are bounded, and without loss of generality scaled to lie in $[0,1]$.  
Each agent $i$'s valuation $\vali$ is drawn independently from a publicly known distribution $\disti$.  
We write $\dists = \disti[1] \times \cdots \times \disti[n]$ for the product distribution over the set $V = V_1 \times \cdots \times V_n$ of valuation profiles.  
We often suppress dependence on $\dists$ from our notation when clear from context.
Agent utilities are quasilinear: if agent $i$ receives outcome $\alloci$ and makes a payment $\pi_i$, his utility is $\utili = \vali(\alloci) - \pi_i$.


The \emph{welfare} of outcome $\allocs$ is $\vals(\allocs) = \sum_i \vali(\alloci)$.  An outcome rule $\alg$ maps each valuation profile to a feasible outcome.  $\alg_i(\vals)$ denotes the outcome of agent $i$ on input $\vals$.  We will write $\OPT(\vals, \feas) = \argmax_{\bx \in \feas}\{\bv(\bx)\}$ for the welfare-maximizing outcome rule for $\feas$, omitting the dependence on $\feas$ when it is clear from context.  

\paragraph{Pricing Rules and Mechanisms}
A \emph{pricing rule} is a profile of functions $\bp = (p_1, \dotsc, p_n)$ that assign prices to outcomes. We write $p_i(x_i \mid \by)$ for the (non-negative) price assigned to outcome $x_i \in X_i$, offered to agent $i$, given partial allocation $\by \in \feas$.  Define $p_i(\alloci) = p_i(\alloci \mid \mathbf{\emptyset})$ for convenience.  We require that $p_i(\alloci \mid \by) = \infty$ for any $\alloci$ such that $(\alloci, \by_{-i}) \not\in \feas$.  
A pricing rule is said to be \emph{monotone non-decreasing} if $p_i(\alloci \mid \by) \geq p_i(\alloci \mid \by_S)$ for all $i$, $\alloci \in X_i$, $\by \in X$, $(\alloci, \by_{-i}) \in \feas$, and $S \subseteq N$.
In general, we allow prices to be dynamic and discriminatory.
We refer to prices that do not depend on the partial allocation (apart from feasibility) as \emph{static} and to prices that do not depend on the identity of the agent as \emph{anonymous}.


A \emph{posted-price mechanism} is defined by a pricing rule $\bp$ and an ordering over the agents.  This pricing rule can, in general, depend on the distributions $\dists$.  The agents are approached sequentially.  Each agent $i$ is presented the menu of prices determined by $p_i$, given all previous allocations, and selects a utility-maximizing outcome.  
A posted-price mechanism is \emph{order-oblivious} if it does not require the agents to be processed in a specific order.  In all of the applications we consider, the mechanisms we construct are order-oblivious.
It is well-known that every posted-price mechanism is \emph{truthful}~\cite{ChawlaHMS10}.  

\paragraph{Online Allocations and Prophet Inequalities}
We consider stochastic allocation algorithms that can depend on the value distributions $\dists$.  That is, an allocation algorithm $\mathcal{A}$ maps a value profile and distribution to a feasible outcome.  We say $\mathcal{A}$ is an online allocation algorithm if $\mathcal{A}_i(\vals, \dists)$ does not depend on the entries of $\vals$ that occur after $i$ in some ordering over the indices.
Extending the notion of competitive ratio from the worst-case analysis of online algorithms, we'll say the \emph{(stochastic) competitive ratio} of online allocation algorithm $\mathcal{A}$ is 
\[ \max_{\dists} \frac{\E_{\vals \sim \dists}[\vals(\OPT(\vals))]}{\E_{\vals \sim \dists}[\vals(\mathcal{A}(\vals, \dists))]}. \]
We somtimes refer to a competitive ratio using its inverse, when convenient.  A \emph{prophet inequality} for constraint $\feas$ is an upper bound on the stochastic competitive ratio of an online allocation algorithm for $\feas$.
We note that a posted-price mechanism describes a particular form of an online allocation algorithm.

\section{A Framework for Prophet Inequalities}
\label{sec:extension}

In this section we state and prove our main result, which reduces prophet inequalities to finding balanced prices for the simpler full information setting.
%
%
We say that a set of outcome profiles $\mathcal{H} \subseteq X$ is \emph{exchange compatible} with $\allocs \in \feas$ if for all $\by \in \mathcal{H}$ and all $i \in N$, $(y_i, \bx_{-i}) \in \feas$.
We call a family of sets $(\feasx)_{\bx \in X}$ \emph{exchange compatible} if $\feasx$ is exchange compatible with $\bx$ for all $\bx \in X$.


\begin{definition}\label{def:ab}
Let $\alpha > 0$, $\beta \ge 0$. Given a set of feasible outcomes $\F$ and a valuation profile $\bv$, a pricing rule $\bp$ is \emph{\ab} with respect to an allocation rule $\ALG$, an exchange-compatible family of sets $(\feasx)_{\bx \in X}$, and an indexing of the players $i = 1, \dots, n$ if for all $\bx \in \feas$
\begin{enumerate}[(a)]
\item \label{conda} $\sum_{i \in N} p_i(x_i \mid \bx_{[i-1]}) \geq \frac{1}{\alpha} \cdot \big( \bv(\ALG(\bv)) -  \bv(\OPT(\bv, \feasx)\big)$, and
\item \label{condb} for all $\bx' \in \feasx$: $\sum_{i \in N} p_i(x'_i \mid \bx_{[i-1]}) \leq \beta \cdot \bv(\OPT(\bv, \feasx))$.
\end{enumerate}
\end{definition}

The definition provides flexibility in the precise choice of $\feasx$.  As $\feasx$ becomes larger (more permissive), both inequalities become easier to satisfy since $\bv(\OPT(\bv, \feasx))$ increases.  On the other hand, a larger set $\feasx$ means that the second condition must be satisfied for more outcomes $\bx' \in \feasx$.
We say that a collection of pricing rules $(p^\bv)_{\bv \in V}$ is $(\alpha,\beta)$-balanced if there exists a choice of $(\feas_\bx)_{\bx\in X}$ such that, for each $\bv$, the pricing rule $p^\bv$ is balanced with respect to $(\feas_\bx)_{\bx \in X}$.

The definition of \AB captures sufficient conditions for a posted-price mechanism to guarantee high welfare when agents have a known valuation profile $\bv$.
Our interest in \ab pricing rules comes from the fact that this result extends to Bayesian settings.

\begin{theorem}\label{thm:main}
Suppose that the collection of pricing rules $(\bp^\bv)_{\bv \in V}$ for feasible outcomes $\feas$ and valuation profiles $\bv \in V$ is \ab with respect to allocation rule $\ALG$ and indexing of the players $i =1, \dots, n$.
Then for $\delta = \tfrac{\alpha}{1+ \alpha \beta}$ the posted-price mechanism with pricing rule $\delta \bp$, where $p_i( x_i \mid \by ) = \E_{\btv}[p^{\btv}_i(x_i \mid \by)]$, generates welfare at least $\tfrac{1}{1+\alpha\beta} \cdot \E_{\bv}[\bv(\ALG(\bv))]$ when approaching players in the order they are indexed.
\label{thm:extension}
\end{theorem}



\begin{proof}
We denote the exchange-compatible family of sets with respect to which the collection of pricing rules $(\bp^\bv)_{\bv\in V}$ is balanced by $(\feas_\bx)_{\bx \in X}$.
We will first use Property (\ref{condb}) to show a lower bound on the utilities of the players, and Property (\ref{conda}) to show a lower bound on the revenue of the posted-price mechanism.  We will then add these together to obtain a bound on the social welfare.

We will write $\bx(\bv)$ for the allocation returned by the posted-price mechanism on input valuation profile $\bv$ and
$\bx'(\bv, \bv') = \OPT( \bv' , \feas_{\bx(\bv)})$ for the welfare-maximizing allocation with respect to valuation profile $\bv'$ under feasibility constraint  $\feas_{\bx(\bv)}$.

\medskip
\noindent{\bf Utility bound:}
We obtain a lower bound on the expected utility of a player as follows. We sample valuations $\bv' \sim \dist$. Player $i$ now considers buying
$\OPT_i( (v_i,\bv'_{-i}) , \feas_{\bx(v'_i,\bv_{-i})})$ at price $\delta{\cdot}p_i( \OPT_i( (v_i,\bv'_{-i}) , \feas_{\bx(v'_i,\bv_{-i})}) \mid \bx_{[i-1]}(\bv))$.
Taking expectations and exploiting that $\bx_{[i-1]}(\bv)$ does not depend on $v_i$ we obtain
\begin{align*}
	\E_{\bv}[u_i(\bv)]
	&\geq \E_{\bv,\bv'} \left[ v_i\left(\OPT_i( (v_i,\bv'_{-i}) , \feas_{\bx(v'_i,\bv_{-i})})\right) - \delta \cdot p_i\Big( \OPT_i( (v_i,\bv'_{-i}) , \feas_{\bx(v'_i,\bv_{-i})}) \growingmid \bx_{[i-1]}(\bv)\Big) \right]\\
	&= \E_{\bv,\bv'} \left[ v'_i\Big(x_i'(\bv, \bv')\Big) - \delta \cdot p_i\Big(x_i'(\bv, \bv') \growingmid \bx_{[i-1]}(\bv)\Big)\right].
\end{align*}
Summing the previous inequality over all agents we get
\begin{align}
	\E_{\bv}\left[\sum_{i \in N} u_i(\bv)\right]
	&\geq \E_{\bv,\bv'} \left[  \sum_{i \in N} v'_i\Big(x_i'(\bv, \bv')\Big) \right] - \E_{\bv,\bv'} \left[ \sum_{i \in N} \delta \cdot p_i\bigg(x_i'(\bv, \bv') ~\bigg|~ \bx_{[i-1]}(\bv)\bigg) \right] \notag\\
	&= \E_{\bv,\bv'} \left[ \bv'\Big(\OPT(\bv' , \feas_{\bx(\bv)})\Big) \right] - \E_{\bv,\bv'} \left[ \sum_{i \in N} \delta \cdot p_i\bigg(x_i'(\bv, \bv') ~\bigg|~ \bx_{[i-1]}(\bv)\bigg) \right]. \label{eq:ext1}
\end{align}
We can upper bound the last term in the previous inequality by using Property (\ref{condb}). This gives
\[
	\sum_{i \in N} \delta \cdot p_i\Big(x_i'(\bv, \bv') ~\Big|~ \bx_{[i-1]}(\bv)\Big)
	\leq \delta\beta \cdot \E_{\btv} \left[ \btv\Big(\OPT(\btv , \feas_{\bx(\bv)} )\Big) \right]
\]
pointwise for any $\bv$ and $\bv'$, and therefore also
\begin{align}
	\E_{\bv,\bv'} \left[ \sum_{i \in N} \delta \cdot p_i\Big(x_i'(\bv, \bv') ~\Big|~ \bx_{[i-1]}(\bv)\Big) \right]
	&\leq \delta\beta \cdot \E_{\bv,\btv} \left[ \btv\Big(\OPT(\btv , \feas_{\bx(\bv)})\Big) \right]. \label{eq:ext2}
\end{align}
Replacing $\bv'$ with $\btv$ in Inequality (\ref{eq:ext1}) and combining it with Inequality (\ref{eq:ext2}) we obtain
\begin{align}
	\E_{\bv}\left[\sum_{i \in N} u_i(\bv)\right]
	\geq (1 - \delta\beta) \cdot \E_{\bv,\btv} \left[ \btv\Big(\OPT(\btv , \feas_{\bx(\bv)})\Big) \right]. \label{eq:ext3}
\end{align}

\medskip
\noindent
{\bf Revenue bound:}
The second step is a lower bound on the revenue achieved by the posted-price mechanism. Applying Property (\ref{conda}) we obtain
\begin{align*}
	\sum_{i \in N} \delta\cdot p_i(x_i(\bv) \mid \bx_{[i-1]}(\bv))
	&= \delta \cdot \sum_{i \in N} \E_{\btv}\left[p_i^{\btv}(x_i(\bv) \mid \bx_{[i-1]}(\bv))\right]\\
	&\ge \frac{\delta}{\alpha} \cdot \E_{\btv} \Big[ \btv(\ALG(\btv)) - \btv(\OPT(\btv , \feas_{\bx(\bv)})) \Big].
\end{align*}
Taking expectation over $\bv$ this shows
\begin{align}
	\E_{\bv} \left[\sum_{i \in N} \delta \cdot p_i(x_i(\bv) \mid \bx_{[i-1]}(\bv)) \right]
	&\geq \frac{\delta}{\alpha} \cdot \E_{\btv} \left[ \btv(\ALG(\btv)) \right] - \frac{\delta}{\alpha} \cdot \E_{\btv,\bv} \left[ \btv(\OPT(\btv, \feas_{\bx(\bv)})) \right]. \label{eq:ext4}
\end{align}

\medskip
\noindent
{\bf Combination:}
It remains to show how the two bounds can be combined so that they imply the approximation guarantee. By quasi-linearity we can rewrite the expected social welfare that is achieved by the posted-price mechanism as the sum of the expected utilities plus the expected revenue.
Using $\delta = \alpha / (1 + \alpha \beta)$ and Inequalities (\ref{eq:ext3}) and (\ref{eq:ext4}), this gives
\begin{align*}
	\E_{\bv}\left[\sum_{i \in N} v_i(x_i(\bv))\right]
	&\geq \E_{\bv}\left[\sum_{i \in N} u_i(\bv)\right] + \E_{\bv}\left[\sum_{i \in N} \delta \cdot p_i\left(x_i(\bv) \mid \bx_{[i-1]}(\bv)\right)\right]  \\
	&\geq (1 - \delta \beta) \E_{\bv,\btv} \left[ \btv(\OPT(\btv , \feas_{\bx(\bv)})) \right] + \frac{\delta}{\alpha} \E_{\btv} \left[ \btv(\ALG(\btv)) \right] - \frac{\delta}{\alpha} \E_{\btv,\bv} \left[ \btv(\OPT(\btv,  \feas_{\bx(\bv)})) \right] \\
	& = \frac{1}{1 + \alpha \beta} \E_{\btv} \left[ \btv(\ALG(\btv) \right]. \qedhere
\end{align*}
\end{proof}


In what follows, we provide an alternative definition of balancedness, in which Property (\ref{condb}) is refined.
This definition will be useful for some applications, as exemplified in Section \ref{section:applications}.
\begin{definition}\label{def:ab2}
Let $\alpha > 0, \beta_1, \beta_2 \geq 0$. Given a set of feasible outcomes $\feas$ and a valuation profile $\bv$, a pricing rule $\bp$ is \emph{weakly $(\alpha, \beta_1, \beta_2)$-balanced} with respect to allocation rule $\ALG$, an exchange-compatible family of sets $(\feas_\bx)_{\bx \in X}$, and an indexing of the players $i = 1, \dots, n$
if, for all $\bx \in \feas$,
\begin{enumerate}[(a)]
\item $\sum_{i \in N} p_i(x_i \mid \bx_{[i-1]}) \geq \frac{1}{\alpha} \cdot \big( \bv(\ALG(\bv)) -  \bv(\OPT(\bv, \feasx)\big)$, and
\item for all $\bx' \in \feasx$: $\sum_{i \in N} p_i(x'_i \mid \bx_{[i-1]}) \leq \beta_1 \cdot \bv(\OPT(\bv, \feasx)) + \beta_2 \cdot \bv(\ALG(\bv))$.
\end{enumerate}
\end{definition}

The following theorem specifies the refined bound on the welfare that is obtained by \abb pricing rules. Its proof appears in
Appendix \ref{app:extension}.
\begin{theorem}\label{thm:refined-extension}
Suppose that the collection of pricing rules $(\bp^\bv)_{\bv \in V}$ for feasible outcomes $\feas$ and valuation profiles $\bv \in V$ is \abb with respect to allocation $\ALG$ and indexing of the players $i = 1, \dots, n$ with $\beta_1 + \beta_2 \geq \frac{1}{\alpha}$.
Then for $\delta = \tfrac{1}{\beta_1 + \max\{ 2 \beta_2, 1/\alpha\} }$ the posted-price mechanism with pricing rule $\delta \bp$, where $p_i( x_i \mid \by ) = \E_{\btv}[p^{\btv}_i(x_i \mid \by)]$, generates welfare at least $\tfrac{1}{\alpha(2 \beta_1 + 4 \beta_2)} \cdot \E_{\bv}[\bv(\ALG(\bv))]$ when approaching players in the order they are indexed.
\end{theorem}

\section{New and Improved Prophet Inequalities}
\label{section:applications}

We have already argued that our framework unifies and simplifies many of the existing prophet inequality proofs.
In this section we show how it can be used to derive new and improved bounds on the approximation ratio that can be obtained via price-based prophet inequalities.
We highlight two results: the new poly-time $O(d)$-approximation for combinatorial auctions with bundle size at most $d$, and the new poly-time constant-approximation for knapsack problems.
Additional results include combinatorial auctions with MPH-$k$ valuations (see Appendix \ref{app:cas}), $d$-sparse packing integer programs (see Appendix \ref{app:sparse-pip}) and multi-dimensional matroids (where the result follows from the Rota exchange theorem \cite[Lemma 2.7]{LeeSV10} and our composition results, see Appendix \ref{app:matroid}).

\paragraph{Combinatorial Auctions with Bounded Bundle Size}

An existential $O(d)$-approximate price-based prophet inequality is presented in Section \ref{sec:intro.example}.
Combined with the $O(d)$-approximation greedy algorithm for this setting, it gives a poly-time $O(d^2)$-approximate price-based prophet inequality (as shown in \cite{FeldmanGL15}). In what follows we use the flexibility of our framework to work directly with a relaxation of the allocation problem, thereby improving the approximation of the prophet inequality from $O(d^2)$ to $O(d)$.
This is a special case of Theorem~\ref{thm:mphk}, which is proved in Appendix~\ref{app:cas}.

\begin{theorem}
\label{thm:dCA}
For combinatorial auctions where every agent can get at most $d$ items, there exist weakly $(1, 1, d-1)$-balanced item prices that are static, anonymous, and order oblivious.  Moreover, a $(4d-2-\epsilon)$-approximate posted-price mechanism can be computed in $\textsc{poly}(n,m,1/\epsilon)$ demand queries, where $\epsilon$ is an additive error due to sampling.
\end{theorem}
\begin{proof}
Consider the canonical fractional relaxation of the combinatorial auction problem: a feasible allocation is described by values $x_{i,S} \in [0,1]$ for all $i \in N$ and $S \subseteq M$ such that $\sum_S x_{i,S} \leq 1$ for all $i$ and $\sum_{i, S \ni j} x_{i,S} \leq 1$ for all $j \in M$.  Take $\feas$ to be all such fractional allocations, and $\feasx$ to be the set of fractional allocations $\by$ such that $\sum_{i, S \ni j} (x_{i,S} + y_{i,S}) \leq 1$ for all $j \in M$, and $\sum_S y_{i,S} \leq 1$ for all $i$.  As usual, we think of $\feasx$ as the set of allocations that remain feasible given a partial allocation $\bx$.

Consider the following pricing rule for fractional allocations.  Given valuation profile $\bv$, let $\bx^*$ be the welfare-maximizing fractional allocation. Then for each item $j$, set $p_j = \sum_i \sum_{S \ni j} x^*_{i,S} v_i(S)$.  We claim that these prices are $(1, 1, d-1)$-balanced with respect to the 
optimal allocation rule.

For Property~(\ref{conda}), fix some $\bx \in \feas$.  Write $x^j = \sum_{i,S \ni j} x_{i,S}$.  Consider the following allocation $\by \in \feasx$: for each $S$, choose $j_S \in \argmax_{j \in S}\{ x^{j} \}$.  Set $y_{i,S} = (1 - x^{j_S}) \cdot x^*_{i,S}$.  We think of $\by$ as the optimal allocation $\bx^*$ adjusted downward to lie in $\feasx$.  We then have that
\begin{align*}
\bv(\bx^*) - \bv(\by) = \sum_i \sum_S x^{j_S}\cdot x^*_{i,S} \cdot v_i(S) = \sum_j x^{j} \sum_{i, S \colon j = j_S} x^*_{i,S} \cdot v_i(S)
\leq \sum_j x^{j} \cdot p_j = \sum_i p_i(x_i).
\end{align*}
Property~(\ref{conda}) follows since $\bv(\by) \leq \bv(\OPT(\bv, \feasx))$.
For Property~(\ref{condb}), fix $\bx \in \feas$ and $\bx' \in \feasx$.  Then
\begin{align*}
\sum_i p_i(x_i') & \leq \sum_j (1-x^j) p_j = \sum_j (1-x^j) \sum_{i, S \ni j} x_{i,S}^* \cdot v_i(S) = \sum_{i,S} x_{i,S}^* \cdot v_i(S) \sum_{j \in S} (1 - x^j) \\
& = \sum_{i, S} (|S|-1)x_{i,S}^* \cdot v_i(S) + \sum_{i,S} x_{i,S}^* \cdot v_i(S) \cdot \left(1 - \sum_{j \in S} x^j \right).
\end{align*}
The first expression on the RHS is at most $(d-1)\bv(\OPT(\bv))$, since $|S| \leq d$ whenever $x_{i,S}^* > 0$.  For the second expression, note that it is at most the welfare of the allocation $\by$ defined by $y_{i,S} = x^*_{i,S} \cdot (1 - \sum_{j \in S} x^j)^+$.  Moreover, this allocation $\by$ is in $\feasx$.  So the second expression is at most $\bv(\OPT(\bv, \feasx))$, giving Property~(\ref{condb}).

Theorem~\ref{thm:refined-extension}
therefore yields prices that guarantee a $(4d-2)$ approximation for the fractional allocation problem, and an $\epsilon$-approximation to those prices can be computed via sampling.  To complete the proof, note that for every agent $i$, if all previous agents have selected integral outcomes, then agent $i$ also has a utility-maximizing outcome that is integral.  This is because any fractional allocation can be interpreted as a convex combination of integral allocations.  These same prices therefore guarantee a $(4d-2-\epsilon)$ approximation even if the mechanism prohibits non-integral allocations from being purchased.
\end{proof}

The more general Theorem \ref{thm:mphk} also improves the best-known polytime prophet inequality for XOS valuations from $2e/(e-1)$ to $2$ (which is tight \cite{FeldmanGL15}) and for MPH-$k$ valuations it improves the best known polytime bounds from $O(k^2)$ to $O(k)$.

\paragraph{Knapsack}
%
%
In the knapsack allocation problem, there is a single divisible unit of resource and each agent has a private value $v_i \geq 0$ for receiving at least $s_i \geq 0$ units.  Assume for now that $s_i \leq 1/2$ for all $i$.  We allow both $v_i$ and $s_i$ to be private information, drawn from a joint distribution.  In our notation: $X_i = [0, \frac{1}{2}]$, $\feas = \{ \bx \mid \sum_i x_i \leq 1\}$, and $v_i(x_i) = v_i$ if $x_i \geq s_i$ and $v_i(x_i) = 0$ otherwise.  Based on an arbitrary allocation algorithm $\ALG$, we design anonymous, static prices by setting $p_i(x_i \mid \by) = x_i \cdot \bv(\ALG(\bv))$ if $x_i$ can feasibly be added and $\infty$ otherwise.  The following restates the second half of Theorem~\ref{thm:knapsack}.

\begin{theorem}
For the knapsack allocation problem in which no single agent can request more than half of the total capacity, the prices above are $(1, 2)$-balanced with respect to $\ALG$.  This implies a $(3+\epsilon)$-approximate polytime posted-price mechanism with a single static anonymous per-unit price.
\end{theorem}
\begin{proof}
The polytime claim follows from Theorem~\ref{thm:extension} with $\ALG$ set to the classic FPTAS for knapsack~\cite{IbarraK75}, so it suffices to prove balancedness.
For any $\bx \in \feas$, let
$\feasx = \feas$
if $\sum_i x_i < \frac{1}{2}$, and $\feasx = \emptyset$ otherwise.  Note that $\feasx$ is exchange compatible with $\bx$ since, for any $\bx' \in \feasx$ and any agent $k$, $x'_k + \sum_i x_i \leq 1$.
To establish balancedness with respect to $(\feasx)_{\bx}$, we consider two cases based on the value of $\sum_i x_i$.

\medskip
\noindent
\textbf{Case 1: $\sum_i x_i < \frac{1}{2}$. } Property~(\ref{conda}) is trivially fulfilled because $\bv(\ALG(\bv)) - \bv(\OPT(\bv, \feasx)) \leq \bv(\OPT(\bv)) - \bv(\OPT(\bv, \feasx)) = 0$. For Property~\ref{condb}, note that for any $\bx' \in \feasx$, we have
\[ \sum_i p_i(x_i' \mid \bx_{[i-1]}) = \sum_i x_i' \cdot \bv(\ALG(\bv)) \leq \bv(\ALG(\bv)) \leq \bv(\OPT(\bv)) = \bv(\OPT(\bv, \feasx)).\]

\noindent
\textbf{Case 2: $\sum_i x_i \geq \frac{1}{2}$. }  Property~(\ref{condb}) is vacuous since $\feasx = \emptyset$. For Property~(\ref{conda}), we have
\[ \sum_i p_i(x_i \mid \bx_{[i-1]}) = \sum_i x_i \cdot \bv(\ALG(\bv)) \geq \frac{1}{2} \bv(\ALG(\bv)) = \frac{1}{2} (\bv(\ALG(\bv)) - \bv(\OPT(\bv, \feasx))). \qedhere\]
\end{proof}

We can remove the restriction that $s_i \leq 1/2$ as follows, completing the proof of Theorem~\ref{thm:knapsack}.  Consider the contribution to the expected optimal welfare separated into welfare from agents with $s_i \leq 1/2$, and agents with $s_i > 1/2$.  The posted-price mechanism described above obtains a $3$-approximation to the former.  For the latter, a mechanism that treats the unit of resource as indivisible, and posts the best take-it-or-leave-it price for the entire unit, is a $2$-approximation.  This is because at most one agent with $s_i > 1/2$ can win in any realization.
Thus, for any distribution profile, one of these two mechanisms must be a $5$-approximation to the unrestricted knapsack problem.\footnote{The worst case is when both mechanisms achieve the same expected welfare, which occurs if $3/5$ of the expected welfare is due to agents with $s_i \leq 1/2$.  The expected welfare of each mechanism is then $\tfrac{1}{3} \cdot \tfrac{3}{5} = \tfrac{1}{5}$ of the optimum.}  One can therefore obtain a $(5+\epsilon)$-approximate price-based prophet inequality by estimating the expected welfare of each pricing scheme (via sampling) and selecting the better of the two.
In Appendix \ref{app:sparse-pip} we show how to generalize the result for the knapsack problem to $d$-sparse packing integer programs.

Finally, consider the fractional version of the knapsack problem, where agents obtain partial value for receiving a portion of their desired allocation: $v_i(x_i) = v_i \cdot \min\{x_i/s_i, 1\}$.  If we restrict allocations $x_i$ to be multiples of some $\delta > 0$, this is a special case of a submodular combinatorial auction with $\lceil 1/\delta \rceil$ identical items.  Since Theorem~\ref{thm:mphk} implies that a fixed per-item price yields a $2$-approximation for any $\delta$, we can infer by taking the limit as $\delta \to 0$ that for any $\epsilon > 0$ there is a $(2+\epsilon)$-approximate polytime posted-price mechanism for the fractional knapsack problem, with a single static anonymous per-unit price, even if each agent's size $s_i$ is private and arbitrarily correlated with their value.  As mentioned in Section~\ref{sec:intro.new}, this improves the previously best-known prophet inequality of $\approx 11.657$ due to~\cite{FeldmanSZ15}.

\newcommand{\smoothoutcome}{x'}
\newcommand{\mechallocrule}{\ensuremath f}
\newcommand{\contraction}{\mathcal{F}/\bx}

\section{From Price of Anarchy to Prophet Inequalities}
\label{sec:smoothness}
In this section we explore the connection between balanced prices and mechanism smoothness. While generally smoothness does not suffice to conclude the existence of a posted-price mechanism with comparable welfare guarantee (see Appendix \ref{app:smoothness-is-weaker}), we will show that this is the case for typical smoothness proofs and present pretty general reductions from the problem of proving prophet inequalities to mechanism smoothness.

We first recall the definition of a smooth mechanism.
A (possibly indirect) mechanism $\mathcal{M}_\pi$ for an allocation problem $\pi$ is defined by a bid space $B = B_1 \times \dots \times B_n$, an allocation rule $\mechallocrule: B \rightarrow \feas$, and a payment rule $P: B \rightarrow \mathbb{R}^n_{\geq 0}$.  We focus on first-price mechanisms, where $P_i(\bb) = b_i(f(\bb))$.  Typically, mechanisms are defined for a collection of problems $\Pi$, in which case we will simply refer to the mechanism as $\mathcal{M}$.
\begin{definition}[\citet{SyrgkanisT13}]\label{def:smoothness}
Mechanism $\mathcal{M}_\pi$ is \emph{$(\lambda,\mu)$-smooth} for $\lambda,\mu \geq 0$ if for any valuation profile $\vals \in V$ and any bid profile $\bids \in B$ there exists a bid $b'_i(\bv,b_i) \in B_i$ for each player $i \in N$ such that
\[
	\sum_{i \in N} u_i(\bidiprime,\bidsmi) \geq \lambda \cdot \vals(\OPT(\vals)) - \mu \cdot \sum_{i \in N} P_i(\bids).
\]
\end{definition}

A mechanism $\mathcal{M}$ that is $(\lambda,\mu)$-smooth has a Price of Anarchy (with respect to correlated and Bayes-Nash equilibria) of at most $\max\{\mu,1\}/\lambda$~\cite{SyrgkanisT13}.

The following formal notion of a residual market will be useful for our further analysis.
For any $\bx \in \feas$ we define the \emph{contraction} of $\feas$ by $\bx$, $\feas/\bx$, as follows. Let $N^+(\bx) = \{i \in N \mid x_i \neq \emptyset\}$. Then $\feas/\bx = \{ \bz = (z_j)_{j \in N\setminus N^+(\bx)} \mid (\bz, \bx_{N^+(\bx)}) \in \feas\}$.
That is, $\feas/\bx$ contains allocations to players who were allocated nothing in $\bx$, that remain feasible when combined with the allocations in $\bx$. 
We think of the contraction by $\bx$  as a \emph{subinstance} on players $N\setminus N^+(\bx)$ with feasibility constraint $\feas/\bx$, and refer to it as the subinstance induced by $\bx$. We say that a collection of problems $\Pi$ is \emph{subinstance closed} if for every $\pi \in \Pi$ with feasible allocations $\feas$ and every $\bx \in \feas$ the subinstance induced by $\bx$ is contained in $\Pi$. The contraction by $\bx$ also naturally leads to an exchange feasible set $\feasx$ by padding the allocations $\bz \in \feas/\bx$ with null outcomes. We refer to this $\feasx$ as the \emph{canonical exchange-feasible set}.


\subsection{Warm-up: Binary, Single-Parameter Problems with Monotone Prices}

We begin with a simple result that serves to illustrate the connection between balancedness and smoothness.  We will show that if a binary, single-parameter problem has the property that the welfare-maximizing mechanism is $(\lambda, \mu)$-smooth 
and its critical prices $\tau_i(\, \cdot\, \mid \by )$ are non-decreasing in $\by$,\footnote{The critical price $\tau_i( \valsmi \mid \by )$ is the infimum of values $v_i$ such that the mechanism allocates $1$ to agent $i$ on input $(v_i, \valsmi)$, in the problem subinstance induced by $\by$.} then there exists a pricing rule that is $(\alpha, \beta)$-balanced, where $\alpha\beta = O(\max\{\mu, 1\} / \lambda)$.  In particular, this implies that the welfare guarantee due to Theorem~\ref{thm:extension} is within a constant factor of the Price of Anarchy of the mechanism implied by smoothness.

\begin{theorem}\label{thm:smooth-warm-up}
Consider a subinstance-closed collection of binary, single-parameter problems such that the first-price mechanism based on the welfare maximizing allocation rule $\OPT$ is $(\lambda, \mu)$-smooth. If the critical prices $\tau_i(\,\cdot\, \mid \by)$ are non-decreasing in $\by$ then setting $p_i(1 \mid \by) = \max\{ v_i, \tau_i(\valsmi \mid \by) \}$ and $p_i(0 \mid \by) = 0$ is $(1, \frac{\mu + 1 + \lambda}{\lambda})$-balanced with respect to $\OPT$ and the canonical exchange-feasible sets $(\feas_\bx)_{\bx \in X}$.
\end{theorem}

\begin{proof}
Fix any $\by$ and $\bx \in \feas_\by$.  Observe that by definition of the prices, it holds that
\begin{equation}
\label{eq:price-lb}
p_i(x_i \mid \by) \geq \bv(\OPT(\bv, \feas_{(\emptyset, \by_{-i})})) - \bv(\OPT(\bv, \feas_{(x_i, \by_{-i})})).
\end{equation}
To see this, first note that both sides of the inequality are equal to $0$ if $x_i = 0$.  If $x_i = 1$ and $v_i \geq \tau_i(\valsmi \mid \by)$, then agent $i$ is allocated in $\OPT(\bv, \feas_{(\emptyset,\by_{-i})})$ and hence both sides of the inequality are equal to $v_i$.  If $x_i = 1$ and $v_i < \tau_i(\valsmi \mid \by)$, then agent $i$ is not allocated in $\OPT(\bv, \feas_{(\emptyset,\by_{-i})})$, and hence the right-hand side of the inequality is at most the externality imposed by forcing an allocation to agent $i$, which is at most $\tau_i(\valsmi \mid \by) = p_i(x_i \mid \by)$.

We are now ready to prove balancedness.  To verify Condition (a), choose $\bx \in \feas$ and note that
\[
\sum_{i=1}^{n} p_i(x_i \mid \bx_{[i-1]}) \geq \sum_{i=1}^{n} \left( \bv(\OPT(\bv, \feas_{\bx_{[i-1]}})) - \bv(\OPT(\bv, \feas_{\bx_{[i]}})) \right) = \bv(\OPT(\bv)) - \bv(\OPT(\bv, \feas_\bx))
\]
as required, where the inequality follows from Equation (\ref{eq:price-lb}), and the equality follows by a telescoping sum.  
For Condition (b), we get
\begin{equation}
\label{eq:critical}
\sum_{i \in \bx'} \tau_i(\valsmi \mid \bx_{[i-1]}) \leq \sum_{i \in \bx'} \tau_i(\valsmi \mid \bx) \leq \frac{\mu + 1}{\lambda} \bv(\OPT(\bv, \feas_\bx)),
\end{equation}
where the first inequality follows by the monotonicity of critical prices, and the second inequality follows by a known implication of smoothness~\cite{DuettingK15} (see Appendix~\ref{app:permeability}).
Therefore, for any $\bx' \in \feas_\bx$,
\begin{align*}
\sum_i p_i(x_i' \mid \bx_{[i-1]}) & \leq \sum_{i \in \bx'} v_i + \sum_{i \in \bx'} \tau_i(\valsmi \mid \bx_{[i-1]}) \displaybreak[0]\\
& \leq \bv(\bx') + \frac{\mu + 1}{\lambda} \bv(\OPT(\bv, \feas_\bx)) \displaybreak[0]\\
& \leq \frac{\mu + 1 + \lambda}{\lambda} \bv(\OPT(\bv, \feas_\bx)),
\end{align*}
where the first inequality follows by replacing the maximum in the definition of the prices by a sum, 
the second inequality follows by Equation (\ref{eq:critical}), 
and the last inequality follows by $\bv(\bx') \leq \bv(\OPT(\bv, \feas_\bx))$, since $\bx' \in \feas_\bx$.
\end{proof}

\subsection{General Problems and Outcome Smoothness}
\label{sec:smoothness-general}

We proceed to show an implication from smoothness to prices that works in more general settings. It is based on the observation that many smoothness proofs proceed by showing that agent $i$ could bid $b'_i$ to get some target outcome $x^*_i$. We capture proofs that proceed in this manner through the following notion of outcome smoothness. Similar but different notions were considered in \cite{DaskalakisS16,LykourisST16}.

\begin{definition}\label{def:outcome-smooth}
A mechanism is \emph{$(\lambda,\mu)$-outcome smooth} for $\lambda,\mu \geq 0$ if for all valuation profiles $\vals \in V$ there exists an outcome $\allocsprime(\vals) \in \feas$ such that for all bid profiles $\bids \in B$,
\[
	\sum_{i \in N} \left( v_i(\allociprime) - \inf_{\bidiprime:\; \mechallocrule_i(\bidiprime,\bidsmi) \succeq \allociprime} P_i(\bidiprime,\bidsmi) \right) \geq \lambda \cdot \vals(\OPT(\vals)) - \mu \cdot \sum_{i \in N} P_i(\bids).
\]
\end{definition}

%
%

We show that if a first-price, declared welfare maximizing mechanism (i.e., a mechanism with allocation rule $f(\bb) = \OPT(\bb)$) is $(\lambda, \mu)$-outcome smooth and has non-decreasing critical prices, then the critical prices for that mechanism (from the definition of outcome smoothness) can be used as posted prices that yield an $O(\lambda/\mu)$ approximation to the optimal welfare.  Recall that these critical prices are different from the first-price payments that make up the mechanism's payment rule.  This result has a mild technical caveat: we require that the mechanism continues to be smooth in a modified problem with multiple copies of each bidders. An allocation is feasible in the modified feasibility space $\feas'$ if it corresponds to a feasible allocation $\bx \in \feas$, with each $x_i$ being partitioned between the copies of agent $i$.


\begin{theorem}
\label{thm:smoothness-to-prices}
Fix valuation space $V$ and feasibility space $\feas$, and suppose $\feas'$ is an extension of $\feas$ as defined above.  Suppose that the first-price mechanism based on the declared welfare maximizing allocation rule for valuation space $V$ and feasibility space $\feas'$ has non-decreasing critical prices, and is $(\lambda, \mu)$-outcome smooth for every $\feas'/\bz$.
Then there is a collection of exchange-feasible sets $(\feas_\bx)_{\bx \in X}$, and an allocation rule $\ALG$ that returns the welfare-maximization allocation with probability $\lambda$, such that for every $\bv \in V$ there exists a pricing rule that is
$(\lambda, \mu/\lambda)$-balanced with respect to $\ALG$ and $(\feas_\bx)_{\bx \in X}$.
\end{theorem}

Theorem~\ref{thm:smoothness-to-prices} implies that posting (an appropriately scaled version of) the critical prices from the outcome smooth mechanism yields a welfare approximation of $O(\lambda/\mu)$, matching the Price of Anarchy guarantee of the original mechanism. The proof of Theorem~\ref{thm:smoothness-to-prices} appears in Appendix~\ref{app:smooth-prices}.

\subsection{Binary, Single-Parameter Problems}
\label{sec:smoothness-binary}


We conclude with two general ``black-box reductions'' for binary single-parameter settings, in which agents can either win or lose, which show how to translate PoA guarantees of $O(\gamma)$ provable via (regular) smoothness into $O(\gamma^2)$-approximate posted-price mechanisms. Proofs appear in Appendix \ref{app:single-parameter}. The key to both these results is a novel, purely combinatorial implication of smoothness for the greedy allocation rule proved in Lemma \ref{lem:zero-one}.

\begin{theorem}\label{thm:greedy}
Suppose that the first-price mechanism based on the greedy allocation rule $\GRD$ has a Price of Anarchy of $O(\gamma)$ provable via smoothness, then there exists a $O(\gamma^2)$-approximate price-based prophet inequality.
\end{theorem}

\begin{theorem}\label{thm:opt}
Suppose that the first-price mechanism based on the declared welfare maximizing allocation rule $\OPT$ has a Price of Anarchy of $O(\gamma)$ provable via smoothness, then there exists a $O(\gamma^2)$-approximate price-based prophet inequality.
\end{theorem}

We note that Theorem \ref{thm:smooth-warm-up} applied to matroids (using known smoothness results for pay-your-bid greedy mechanisms over matroids~\cite{LucierS15}) implies the existence of $(1,3)$-balanced prices and hence a $4$-approximate prophet inequality. A strengthening of Theorem \ref{thm:opt} for monotonically increasing critical prices (discussed in Remark \ref{rem:monotone}) leads to an improved factor of $2$, matching the prophet inequality for matroids shown in \citet{KleinbergW12}.  This also captures the classic single-item prophet inequality, as a special case.

\section{Conclusions and Open Problems}
We introduced a general framework for establishing prophet inequalities and posted price mechanisms for multi-dimensional settings.
This work leaves many questions open.

A general class of questions is to determine the best approximation guarantee that a prophet inequality can achieve for a particular setting.  For example, even for the intersection of two matroids there is a gap between the trivial lower bound of $2$ and the upper bound of $4k-2 = 6$. Similarly, in subadditive combinatorial auctions, the best-known upper bound is logarithmic in the number of items $m$ \cite{FeldmanGL15}, but again the best-known lower bound is $2$, inherited from the case of a single item.  Notably, the price of anarchy for simultaneous single-item auctions is known to be constant for subadditive valuations~\cite{FeldmanFGL13}, but the proof does not use the smoothness framework and hence our results relating posted prices to smooth mechanisms do not directly apply.

A related question is whether there exist prophet inequalities that cannot be implemented using posted prices. Interestingly, we are not aware of any separation between the two so far. More generally, one could ask about the power of anonymous versus personalized prices, item versus bundle prices, static versus dynamic prices, and so on.  For example, to what extent can \emph{static} prices approximate the welfare under a matroid constraint, an intersection of matroids, or an arbitrary downward-closed feasibility constraint?


Regarding the pricing framework itself, it would be interesting to extend the notion of \AB to allow randomization in a dynamic pricing rule, and to understand the additional power of randomization.  One could also generalize beyond feasibility constraints to more general seller-side costs for allocations.  For the connection between smoothness and balancedness, we leave open the question of removing the price-monotonicity condition from Theorem~\ref{thm:smoothness-to-prices}, or whether the approximation factors can be improved for our single-parameter reductions (Theorems~\ref{thm:greedy} and~\ref{thm:opt}). Finally, recent work has shown that smoothness guarantees often improve as markets grow large~\cite{FeldmanILRS16}; is there a corresponding result for balancedness?


\bibliographystyle{abbrvnat}
\bibliography{biblio}

\newpage

\appendix



\section{Classic Prophet Inequality via Balanced Prices}
\label{app:single-item}

In this appendix we show how to re-derive the classic prophet inequality via our framework in Section \ref{sec:extension}. Specifically, we show the existence of a $(1,1)$-balanced pricing rule as defined in Definition \ref{def:ab}. Theorem \ref{thm:main} then shows the factor $2$-approximation.

In the classic setting we have $X_i = \{0, 1\}$ for all $i$ and $\feas = \{\bx \mid \sum_i x_i \leq 1\}$.
We set $p_i(1 \mid \bx) = \max_\ell v_\ell$ if $\bx$ does not allocate the item, $\infty$ otherwise; $p_i(0 \mid \bx) = 0$ for all $\bx$.
This corresponds to a fixed posted price of $\max_\ell v_\ell$ on the item.

\begin{claim}
These prices are $(1, 1)$-balanced with respect to
$\OPT$ and
$\feasx$ defined by $\feasx = \feas$ if $\bx$ does not allocate the item and $\feasx = \emptyset$ otherwise.
\end{claim}
\begin{proof}
Let $\bx$ be an arbitrary allocation profile.
If $\bx$ allocates the item, then $\feasx = \emptyset$ and thus Condition (b) is trivially fulfilled. For Condition (a), we observe that  $\bv(\OPT(\bv, \feasx)) = 0$ and that $\bv(\OPT(\bv)) = \max_\ell v_\ell = \sum_i p_i(x_i \mid \bx_{[i-1]})$, because exactly one buyer pays $\max_\ell v_\ell$. If $\bx$ does not allocate the item, then $\bv(\OPT(\bv, \feasx)) = \bv(\OPT(\bv))$, making Condition (a) trivial. For Condition (b), we use that in $\bx'$ at most one buyer is allocated the item. Therefore, $\sum_i p_i(x_i' \mid \bx_{[i-1]}) \leq \max_\ell v_\ell = \bv(\OPT(\bv)) = \bv(\OPT(\bv, \feasx))$.
\end{proof}

\section{Proof of Theorem \ref{thm:refined-extension}}
\label{app:extension}\label{app:refined-balance}


Our proof will follow the same steps as the one for Theorem~\ref{thm:extension}. As in that poof, denote the exchange-compatible family of sets with respect to which the collection of pricing rules $(\bp^\bv)_{\bv\in V}$ is balanced by $(\feas_\bx)_{\bx \in X}$, and let $\bx(\bv)$ be the allocation returned by the posted-price mechanism on input valuation profile $\bv$.
Let $\bx'(\bv, \bv') = \OPT( \bv', \feas_{\bx(\bv)})$ be the allocation that maximizes welfare with respect to valuation profile $\bv'$ over feasibility constraint $\feas_{\bx(\bv)}$.

\medskip
\noindent{\bf Utility bound:}
Again sample valuations $\bv' \sim \dist$. By the same reasoning as in the proof of Theorem~\ref{thm:extension}, we obtain
\begin{align}
	\E_{\bv}\left[\sum_{i \in N} u_i(\bv)\right]
	= \E_{\bv,\bv'} \left[ \bv'(\bx'(\bv, \bv')) \right] - \E_{\bv,\bv'} \left[ \sum_{i \in N} \delta \cdot p_i\bigg(x_i'(\bv, \bv') ~\bigg|~ \bx_{[i-1]}(\bv)\bigg) \right]. \label{eq:extext1}
\end{align}
We upper bound the last term in the previous inequality by using Property (\ref{condb}). This gives pointwise for any $\bv$ and $\bv'$
\[
	\sum_{i \in N} \delta \cdot p_i\bigg(x_i'(\bv, \bv') ~\bigg|~ \bx_{[i-1]}(\bv)\bigg)
	\leq \delta\beta_1 \cdot \E_{\btv} \left[ \btv(\OPT(\btv, \feas_{\bx(\bv)} ))\right] + \delta \beta_2 \cdot \E_{\btv} \left[ \btv(\ALG(\btv))\right],
\]
and therefore also
\begin{align}
	\E_{\bv,\bv'} \left[ \sum_{i \in N} \delta \cdot p_i\bigg(x_i'(\bv, \bv') ~\bigg|~ \bx_{[i-1]}(\bv)\bigg) \right]
	&\leq \delta\beta_1 \cdot \E_{\bv,\btv} \left[ \btv(\OPT(\btv, \feas_{\bx(\bv)}) \right] + \delta\beta_2 \cdot \E_{\bv,\btv} \left[ \btv(\ALG(\btv) \right]. \label{eq:extext2}
\end{align}
Replacing $\bv'$ with $\btv$ in Inequality (\ref{eq:extext1}) and combining it with Inequality (\ref{eq:extext2}) we obtain
\begin{align}
	\E_{\bv}\left[\sum_{i \in N} u_i(\bv)\right]
	\geq (1 - \delta\beta_1) \cdot \E_{\bv,\btv} \left[ \btv(\OPT(\btv, \feas_{\bx(\bv)}) \right] - \delta\beta_2 \cdot \E_{\bv,\btv} \left[ \btv(\ALG(\btv) \right]. \label{eq:extext3}
\end{align}

\noindent
{\bf Revenue bound:}
Again, applying Property (\ref{conda}) we obtain
\begin{align}
	\E_{\bv} \left[\sum_{i \in N} \delta \cdot p_i(x_i(\bv) \mid \bx_{[i-1]}(\bv)) \right]
	&\geq \frac{\delta}{\alpha} \cdot \E_{\btv} \left[ \btv(\ALG(\btv)) \right] - \frac{\delta}{\alpha} \cdot \E_{\btv,\bv} \left[ \btv(\OPT(\btv, \feas_{\bx(\bv)}) \right]. \label{eq:extext4}
\end{align}

\noindent
{\bf Combination:}
To combine our bounds, we distinguish whether $\beta_2 \geq \frac{1}{2 \alpha}$.

\begin{description}
\item[Case 1: $\beta_2 \geq \frac{1}{2 \alpha}$]
We use that point-wise for every $i \in N$ we have $u_i(\bv) \geq 0$. Therefore, $u_i(\bv) \geq \rho u_i(\bv)$ for all $0 \leq \rho \leq 1$. Using $\delta = \frac{1}{\beta_1 + 2 \beta_2}$, $\rho = \frac{1}{2 \alpha \beta_2}$, and Inequalities (\ref{eq:extext3}) and (\ref{eq:extext4}), we get
\begin{align*}
	\E_{\bv}\left[\sum_{i \in N} v_i(x_i(\bv))\right]
	&\geq \rho \cdot \E_{\bv}\left[\sum_{i \in N} u_i(\bv)\right] + \E_{\bv}\left[\sum_{i \in N} \delta \cdot p_i\left(x_i(\bv) \mid \bx_{[i-1]}(\bv)\right)\right]  \\
	&\geq \rho (1 - \delta\beta_1) \cdot \E_{\bv,\btv} \left[ \btv(\OPT(\btv, \feas_{\bx(\bv)}) \right] - \rho \delta\beta_2 \cdot \E_{\bv,\btv} \left[ \btv(\ALG(\btv) \right] \\
	& \qquad \qquad \qquad \; + \frac{\delta}{\alpha} \cdot \E_{\btv} \left[ \btv(\ALG(\btv)) \right] - \frac{\delta}{\alpha} \cdot \E_{\btv,\bv} \left[ \btv(\OPT(\btv, \feas_{\bx(\bv)}) \right] \\
	& = \frac{1}{\alpha (2 \beta_1 + 4 \beta_2)} \cdot \E_{\btv} \left[ \btv(\ALG(\btv) \right].
\end{align*}
\item[Case 2: $\beta_2 < \frac{1}{2 \alpha}$]
Now, we use $\delta = \frac{1}{\beta_1 + 1/\alpha}$. Inequalities (\ref{eq:extext3}) and (\ref{eq:extext4}) yield
\begin{align*}
	\E_{\bv}\left[\sum_{i \in N} v_i(x_i(\bv))\right]
	&\geq \E_{\bv}\left[\sum_{i \in N} u_i(\bv)\right] + \E_{\bv}\left[\sum_{i \in N} \delta \cdot p_i\left(x_i(\bv) \mid \bx_{[i-1]}(\bv)\right)\right]  \\
	&\geq (1 - \delta\beta_1) \cdot \E_{\bv,\btv} \left[ \btv(\OPT(\btv, \feas_{\bx(\bv)}) \right] - \delta\beta_2 \cdot \E_{\bv,\btv} \left[ \btv(\ALG(\btv) \right] \\
	& \qquad \qquad \qquad \; + \frac{\delta}{\alpha} \cdot \E_{\btv} \left[ \btv(\ALG(\btv)) \right] - \frac{\delta}{\alpha} \cdot \E_{\btv,\bv} \left[ \btv(\OPT(\btv, \feas_{\bx(\bv)}) \right] \\
	& = \frac{1 - \alpha \beta_2}{1 + \alpha \beta_1} \cdot \E_{\btv} \left[ \btv(\ALG(\btv) \right] \\
	& \geq \frac{1}{\alpha (2 \beta_1 + 4 \beta_2)} \E_{\btv} \left[ \btv(\ALG(\btv) \right],
\end{align*}
where the last step uses that $\beta_1 + \beta_2 \geq \frac{1}{\alpha}$ and $\beta_2 < \frac{1}{2 \alpha}$.
\end{description}


\section{Composition Results}
\label{app:composition}

In this section we show that balanced
prices are \emph{composable}, in the sense that balanced prices for separate markets remain balanced when the markets are combined.  We consider two forms of composition. The first is a composition of preferences: it shows how to extend from a class of valuations $V$ to any maximum over valuations from $V$. The second shows how to compose allocations across different markets, where agents have additive preferences across markets. Together these two composition results capture XOS composition, in the sense of~\cite{SyrgkanisT13}.
Our theorems apply to balancedness with respect to $\OPT$, and they extend to general approximation algorithms $\ALG$ under mild conditions.

\paragraph{Closure under Maximum}
Given an arbitrary valuation space $V_i$ for player $i$, we consider its extension $V^{\max}_i$, which contains all functions $v_i^{\max}\colon X_i \to \reals$ for which there is a finite set $\{v_i^1, \ldots, v_i^m\} \in V_i$ such that $v_i^{\max}(x_i) = \max_{\ell} v_i^\ell(x_i)$ for all $x_i \in X_i$. We say that a valuation profile $\btv \in V$ is a \emph{supporting} valuation profile for allocation
$\bx$
and valuation profile $\bv \in V^{max}$ if $\tilde{v}_i \leq v_i$ and
$\tilde{v}_i(x_i) = v_i(x_i)$ for all $i$.

\begin{definition}
\label{def:consistent}
Allocation rule $\ALG$ is \emph{consistent} if for every $\bv \in V^{\max}$ and corresponding supporting valuation profile $\btv \in V$ for $\ALG(\bv)$, we have $\btv(\ALG(\btv)) \geq \btv(\ALG(\bv))$.
\end{definition}
%
%
\begin{lemma}
The optimal allocation rule $\OPT$ is consistent.
\end{lemma}
\begin{proof}
Since $\OPT(\bv)$ is a feasible outcome under $\btv$,
it holds that $\btv(\OPT(\btv)) \geq \btv(\OPT(\bv))$ by optimality.
\end{proof}

\begin{theorem}\label{thm:max}
Suppose that for each $\bv \in V$ there exists a pricing rule $\bp^{\bv}$ that is \ab with respect to $\bv$, consistent allocation rule $\ALG$, and the exchange-compatible family of sets $(\feas_\bx)_{\bx \in X}$. For each $\bv \in V^{\max}$ let $\btv \in V$ be a supporting valuation profile for $\ALG(\bv)$. Then the pricing rule $\bp^{\btv}$ is \ab with respect to $\bv$, $\ALG$, and $(\feas_\bx)_{\bx \in X}$.
\end{theorem}
\begin{proof}
We first establish Property (\ref{conda}). We use Property (\ref{conda}) of the original pricing rules, plus the definition of a supporting valuation profile to conclude that
\begin{align*}
\sum_{i \in N} p_i^{\btv}(x_i \mid \bx_{[i-1]})
& \geq \frac{1}{\alpha} \cdot \left( \btv(\ALG(\btv)) -  \btv(\OPT(\btv, \feasx)\right) \\
& \geq \frac{1}{\alpha} \cdot \big( \btv(\ALG(\bv)) -  \bv(\OPT(\btv, \feasx)\big) \\
& \geq \frac{1}{\alpha} \cdot \big( \bv(\ALG(\bv)) -  \bv(\OPT(\bv, \feasx)\big).
\end{align*}

Next we establish Property (\ref{condb}). By Property (\ref{condb}) of the original pricing rule and the definition of a supporting valuation profile,
\begin{align*}
\sum_{i \in N} p_i^{\btv}(x'_i \mid \bx_{[i-1]})
&\leq \beta \cdot \btv(\OPT(\btv, \feasx))\\
&\leq \beta \cdot \bv(\OPT(\btv, \feasx))\\
&\leq \beta \cdot \bv(\OPT(\bv, \feasx)).  \qedhere
\end{align*}
\end{proof}

\paragraph{Closure under Addition}
Suppose there are $m$ separate allocation problems, each with feasibility constraint $\F^\ell$ over allocation space $X^\ell = X_1^\ell \times \dotsc \times X_n^\ell$.  The joint problem is then defined over the product allocation space $X = X^1 \times \dotsc \times X^m$, with feasibility constraint $\F = \F^1 \times \dotsc \times \F^m$. We say that a valuation $v_i\colon X \to \reals$ is \emph{additive} if it is defined by a (not necessarily additive) valuation function $v_i^\ell$ for each outcome $x_i^\ell \in X_i^\ell$, and for an outcome $x_i =(x_i^1, \dotsc, x_i^m) \in X_i$, the value of $x_i$ is given by $v_i(x_i) = \sum_{\ell = 1}^{m} v_i^\ell(x_i^\ell)$.


\begin{theorem}\label{thm:add}
Suppose that $\bv$ is additive over a set of allocation problems $\F^1, \dots, \F^m$, and for each individual allocation
problem there exists a pricing rule $\bp^{\bv^\ell}$ that is \ab with respect to $\bv^\ell$,
allocation rule $\ALG^\ell$ on $\F^\ell$, and the exchange-compatible family of sets $(\feas^\ell_{\bx^\ell})_{\bx^\ell \in X^\ell}$. Then the
pricing rule $\bp = \sum_{\ell = 1}^{m} \bp^{\bv^\ell}$
is \ab with respect to $\bv$ and
$\ALG$ on $\F$, where $\ALG(\bv) = (\ALG^1(\bv^1), \dots, \ALG^m(\bv^m))$.
\end{theorem}
\begin{proof}
Set $\feas_\bx = (\prod_{i=1}^{m} \feas^\ell_{\bx^\ell})_{\bx \in X}$. We first verify Condition (\ref{conda}). We use the definition of the joint pricing rule, Condition (\ref{conda}) of the component pricing rules together with additivity across subproblems to conclude that
\begin{align*}
\sum_{i \in N} p_i(x_i \mid \bx_{[i-1]}) & = \sum_{i \in N} \sum_{\ell = 1}^{m} p_i^{\bv^\ell}(x_i^\ell \mid \bx_{[i-1]}^\ell) \\
& \geq \sum_{\ell=1}^{m} \frac{1}{\alpha} \cdot \left( \bv^{\ell}(\ALG^\ell(\bv^\ell)) -  \bv^{\ell}(\OPT(\bv^\ell, \feasxell)\right) \\
& = \frac{1}{\alpha} \cdot \left( \bv(\ALG(\bv)) -  \bv(\OPT(\bv, \feasx))\right).
\end{align*}

Next we verify Condition (\ref{condb}). We use the definition of the joint pricing rule, Condition (\ref{condb}) of the component pricing rules together with additivity across subproblems to conclude that
\begin{align*}
\sum_{i \in N} p_i(x'_i \mid \bx_{[i-1]}) &= \sum_{\ell = 1}^{m} \sum_{i \in N} p_i^{\bv^\ell}((x')_i^\ell \mid \bx_{[i-1]}^\ell)\\
& \leq \sum_{\ell = 1}^{m} \beta \cdot \bv^{\ell}(\OPT(\bv^\ell, \feasxell)) \\
& = \beta \cdot \bv(\OPT(\bv, \feasx)). \qedhere
\end{align*}
\end{proof}

We note that both Theorem \ref{thm:max} and Theorem \ref{thm:add} can be generalized so that they apply to weakly balanced pricing rules.

\section{Proof of Theorem \ref{thm:mphk}}
\label{app:cas}

In this appendix we establish our new prophet inequality results for combinatorial auctions with MPH-$k$ valuations. In particular, we establish the existence of $(1,1)$-balanced prices for XOS combinatorial auctions.

\paragraph{Combinatorial Auctions with MPH-$k$ Valuations}

The {\em maximum over positive hypergraphs} (MPH) hierarchy of valuations \cite{FeigeFIILS15} is an inclusive hierarchy (i.e., it is expressive enough to include all valuations), which subsumes many interesting classes of valuations as special cases.

To formalize this valuation class, we first need a few preliminaries.
A hypergraph representation $w$ of valuation function $v\colon 2^M \to \reals_{\geq 0}$ is a set function that satisfies $v(S) = \sum_{T \subseteq S} w(T)$. Any valuation function $v$ admits a unique hypergraph representation and vice versa.
A set $S$ such that $w(S) \neq 0$ is said to be a {\em hyperedge} of $w$.
Pictorially, the hypergraph representation can be thought as a weighted hypergraph, where every vertex is associated with an item in $M$, and the weight of each hyperedge $e\subseteq M$ is $w(e)$. Then the value of the function for any set $S\subseteq M$, is the total value of all hyperedges that are contained in $S$.
The {\em rank} of a hypergraph representation $w$ is the cardinality $k$ of the largest hyperedge.
The rank of $v$ is the rank of its corresponding $w$ and we refer to a valuation function $v$ with rank $k$ as a \emph{hypergraph-$k$} valuation. If the hypergraph representation of $v$ is non-negative, i.e. for any $S\subseteq M$, $w(S)\geq 0$, then we refer to function $v$ as a \emph{positive hypergraph-$k$} function (PH-$k$) \cite{AbrahamBDR12}.

\begin{definition}[Maximum Over Positive Hypergraph-$k$ (MPH-$k$) class \cite{FeigeFIILS15}]
A monotone valuation function $v:2^M\to \reals_{\geq 0}$ is {\em Maximum over Positive Hypergraph-$k$} (MPH-$k$) if it can be expressed as a maximum over a set of PH-$k$ functions.  That is, there exist PH-$k$ functions $\{v_{\ell}\}_{\ell\in\Ell}$ such that for every set $S \subseteq M$,
\begin{equation}
\textstyle{v(S) = \max_{\ell \in \Ell} v_{\ell}(S)},
\end{equation}
where $\Ell$ is an arbitrary index set.
\end{definition}

An important special case of MPH-$k$ are XOS valuations, which are defined as the maximum over additive functions, and therefore coincide with MPH-$1$.

\paragraph{Existential $O(k)$ Result}


Given an arbitrary allocation algorithm $\ALG$ and valuation profile $\bv$, write $\tilde{\bv}$ for the supporting Hypergraph-$k$ valuations for allocation $\ALG(\bv)$. Write $w_i$ for the hypergraph representation of $\tilde{v}_i$, so that $\tilde{v}_i(\ALG(\bv)) = \sum_{T \subseteq \ALG_i(\bv)}w_i(T)$.

Then, for each item $j$, we define
\begin{equation}
\label{eq:prices}
p^{\bv}(\{j\}) = \sum_{\substack{T \ni j\\T \subseteq \ALG_i(\bv)}} w_i(T).
\end{equation}
That is, item prices are determined by adding up the weights of each supporting hyperedges an item is contained in. These prices then extend linearly: For each set of goods $x$, set $p^{\bv}(x) = \sum_{j \in x}p^{\bv}(\{j\})$.  Finally, for each agent $i$ and allocation $x_i$, we will have $p_i^{\bv}(x_i \mid \bz) = p^{\bv}(x_i)$ whenever $x_i$ is disjoint from the sets in $\bz$ and $\infty$ otherwise.

\begin{theorem}
\label{thm.mphk.balanced}
The pricing rule defined in Equation \ref{eq:prices}, extended linearly to sets of items, is weakly $(1, 1, k-1)$-balanced with respect to an arbitrary allocation rule $\ALG$ and MPH-$k$ valuation profile $\vals$.
\end{theorem}

\begin{proof}
We will show balancedness with respect to $\feasx = \{ \by \in \feas \colon (\bigcup_i y_i) \cap (\bigcup_i x_i) = \emptyset \}$. Observe that we can lower-bound the value of $\OPT(\bv, \feasx)$ by removing all items that are allocated by $\bx$ from the allocation $\ALG(\bv)$. This gives us
\begin{equation}
\bv(\OPT(\bv, \feasx)) \geq \sum_{\ell \in N} v_\ell(\ALG_\ell(\bv) \setminus (\bigcup_i x_i)) \geq \sum_\ell \sum_{\substack{T \subseteq \ALG_\ell(\bv)\\\forall i \colon T \cap x_i = \emptyset}}w_\ell(T).
\label{eq:mphkoptfeasx}
\end{equation}

Furthermore, note first that $p^{\bv}$ is a fixed pricing scheme, meaning that the price of an outcome does not change unless it becomes infeasible. Therefore, for every allocation $\by \in \feasx$, we have
\begin{align}
\sum_i p_i(y_i \mid \bx_{[i-1]}) & = \sum_i p_i(y_i \mid \bx) \nonumber\\
& = \sum_i \sum_{j \in y_i} p^{\bv}(\{j\}) \nonumber\\
& = \sum_i \sum_\ell \sum_{j \in \ALG_\ell(\bv) \cap y_i} p^{\bv}(\{j\})\nonumber\\
& = \sum_\ell \sum_i \sum_{j \in \ALG_\ell(\bv) \cap y_i} \sum_{\substack{T \ni j\\T \subseteq \ALG_\ell(\bv)}} w_\ell(T).
\label{eq:mphkprices}
\end{align}

For condition (a), by Equation~\eqref{eq:mphkprices},
\begin{align*}
\sum_i p_i^{\bv}(x_i \mid \bx_{[i-1]}) & \geq \sum_\ell \sum_{\substack{T \subseteq \ALG_\ell(\bv)\\\exists i \colon T \cap x_i \neq \emptyset}}w_\ell(T) \\
& = \bigg(\sum_\ell \sum_{T \subseteq \ALG_\ell(\bv)}w_\ell(T) - \sum_\ell \sum_{\substack{T \subseteq \ALG_\ell(\bv)\\\forall i \colon T \cap x_i = \emptyset}}w_\ell(T)\bigg) \\
& \geq (\bv(\ALG(\bv)) -  \bv(\OPT(\bv, \feasx))),
\end{align*}
where the first inequality follows by changing the order of summation over $j$ and $T$, and the second inequality is given in Equation~\ref{eq:mphkoptfeasx}.

For condition (b), note that for all $\allocs$ and all $\bx' \in \feasx$, by Equation~\eqref{eq:mphkprices}, after splitting the sum depending on whether the respective set $T$ intersects with any of the bundles in $\allocs$ or not
\begin{align*}
\sum_i p_i(x_i' \mid \bx_{[i-1]}) & = \sum_\ell \sum_i \sum_{j \in \ALG_\ell(\bv) \cap x_i'} \sum_{\substack{T \ni j\\T \subseteq \ALG_\ell(\bv) \\ T \cap(\bigcup_{i'} x_{i'}) = \emptyset}} w_\ell(T) + \sum_\ell \sum_i \sum_{j \in \ALG_\ell(\bv) \cap x_i'} \sum_{\substack{T \ni j\\T \subseteq \ALG_\ell(\bv) \\ T \cap(\bigcup_{i'} x_{i'}) \neq \emptyset}} w_\ell(T).
\end{align*}
Observe that in the first sum for a fixed set $T$, the term $w_\ell(T)$ occurs at most $\lvert T \rvert$ times. In the second sum, it can even occur only $\lvert T \rvert - 1$ times because the intersection of $\allocs$ and $\allocs'$ is empty but $\allocs$ intersects with $T$. By applying that $\lvert T \rvert \leq k$ whenever $w_\ell(T) > 0$, this gives us
\begin{align*}
\sum_i p_i(x_i' \mid \bx_{[i-1]})& \leq \sum_\ell \sum_{\substack{T \subseteq \ALG_\ell(\bv) \\ T \cap(\bigcup_{i'} x_{i'}) = \emptyset}} \lvert T \rvert w_\ell(T) + \sum_\ell \sum_{\substack{T \subseteq \ALG_\ell(\bv) \\ T \cap(\bigcup_{i'} x_{i'}) \neq \emptyset}} (\lvert T \rvert - 1) w_\ell(T) \\
& \leq k \sum_\ell \sum_{\substack{T \subseteq \ALG_\ell(\bv) \\ T \cap(\bigcup_{i'} x_{i'}) = \emptyset}} w_\ell(T) + (k-1) \sum_\ell \sum_{\substack{T \subseteq \ALG_\ell(\bv) \\ T \cap(\bigcup_{i'} x_{i'}) \neq \emptyset}} w_\ell(T) \\
& = \sum_\ell \sum_{\substack{T \subseteq \ALG_\ell(\bv) \\ T \cap(\bigcup_{i'} x_{i'}) = \emptyset}} w_\ell(T) + (k-1) \sum_\ell \sum_{T \subseteq \ALG_\ell(\bv)} w_\ell(T) \\
& \geq \bv(\OPT(\bv, \feasx)) + (k-1) \bv(\ALG(\bv)). \qedhere
\end{align*}
\end{proof}

Note that Theorem \ref{thm.mphk.balanced} when specialized to XOS valuations shows the existence of weakly $(1,1,0)$-balanced prices, or equivalently, $(1,1)$-balanced prices.


\paragraph{Computational $O(k)$ Result}

We now show how to obtain a \emph{polytime} $(4k-2)$-approximate price-based prophet inequality for MPH-$k$ valuations.
For this result we will assume access to the following kind of \emph{MPH-$k$ oracle}. Suppose that valuation function $v$ is MPH-$k$, with supporting PH-$k$ functions $\{v_{\ell}\}_{\ell \in L}$. The query for $v$ takes as input a set of items $S$, and returns a value oracle to access the PH-$k$ function $v_{\ell}$ for which $v(S) = v_{\ell}(S)$. That is, for every $T \subseteq M$, we can query the value of $v_{\ell}(T)$ in a second step.

%

The following linear problem, known as the {\em configuration LP} for combinatorial auctions, computes a fractional allocation that maximizes the social welfare among all fractional allocations.

\begin{eqnarray}
\max & & \sum_{i,S} \vali(S)\cdot \alloc_{_{i,S}} \nonumber\\
\mbox{s.t.} & & \sum_S \alloc_{_{i,S}} \leq 1 \mbox{ for every } i \in N \nonumber\\
& & \sum_{i, S: j \in S} \alloc_{_{i,S}} \leq 1 \mbox{ for every } j \in M \nonumber\\
& & \alloc_{_{i,S}} \in [0,1] \mbox{ for every } i \in N, S \subseteq M \nonumber
\end{eqnarray}

We extend the definition of $\feas$ to include all fractional allocations that fulfill the above LP; $\feasx$ is extended to all fractional allocations that use at most a fractional equivalent of $1 - \sum_{i, S: j \in S} \alloc_{_{i,S}}$ of every item $j \in M$. Our first observation will be that Theorem~\ref{thm.mphk.balanced} holds even for the fractional version of the MPH-$k$ combinatorial auction problem.  That is, if the set of feasible outcomes is extended to include all fractional allocations, then an appropriate extension of the prices described above remain weakly $(1, 1, k-1)$-balanced. In particular, given a fixed valuation profile $\bv$, we define prices based on the optimal LP solution $\bx^\ast$. For every agent $i$ and every bundle $S$, let $w_{i, S}$ be the PH-$k$ representation in the support that maximizes agent $i$'s valuation of $S$, which is $v_i(S) = \sum_{T \subseteq S} w_{i, S}(T)$. Then, price item $j$ as follows:
$p(\{j\}) = \sum_i \sum_S x^\ast_{i,S} \sum_{T: j \in T, T \subseteq S} w_{i, S}(T) = \sum_i \sum_S x^\ast_{i,S} (w_{i, S}(S) - w_{i, S}(S \setminus \{ j \}))$. This sum can be computed in polynomial time given oracle access as described above because $x^\ast_{i,S} \neq 0$ only for polynomially many $S$. Then, extend this pricing linearly to prices over sets of items, and to prices over fractional allocations by simple scaling, i.e., for every $\bx \in \feas$ and $\bx' \in \feasx$, we let $p(x_i' \mid \bx) = \sum_S x_{i, S}' \sum_{j \in S} p(\{j\})$.

It is well known that an optimal fractional solution to the configuration LP can be computed in polynomial time given access to demand queries (using demand queries to implement a separation oracle, as is standard).
Therefore, for the space of fractional allocations, prices that are weakly $(1, 1, k-1)$-balanced with respect to the optimal fractional solution can be computed in polynomial time, given access to demand and MPH-$k$ oracles.

We can therefore compute prices that give a $(4k-2)$-approximation to the optimal fractional allocation, in a posted price mechanism, where agents are free to purchase any fractional allocation at the posted prices.
This, however, seems unsatisfactory. After all, we do not wish to allow agents to purchase infeasible sets, and the analysis only applies if every agent can purchase a set in her demand correspondence. The following observation comes to our help: for every agent $i$, if all previous agents have selected integral outcomes from the posted price mechanism, then agent $i$ has a utility-maximizing outcome in their demand correspondence that is integral.  This is because any fractional allocation can be interpreted as a convex combination of integral allocations.  Since our approximation guarantee holds regardless of the demanded set chosen by each agent, it will hold even if we restrict agents to only select integral outcomes in the posted price mechanism. This means that we only need to price integral allocations and the analysis follows.
We conclude that the prices computed using the configuration LP, as described above, actually generate a $(4k-2)$ approximation for the (non-fractional) MPH-$k$ combinatorial auction problem.


\section{Proof of Theorem \ref{thm:sparse-pip}}
\label{app:sparse-pip}

In this appendix we prove our polytime price-based prophet inequalities for sparse linear packing programs. We consider programs with and without integrality constraints. That is, the possible outcomes for agent $i$ are either $[0, 1]$ (fractional solutions) or $\{0, 1\}$ (integral solutions). We assume that valuations are linear, i.e., $v_i(x_i) = v_i \cdot x_i$. In both cases, the feasibility constraints $\feas$ are given by a constraint matrix $\mathbf{A} \in \reals_{\geq 0}^{m \times n}$ such that $\bx \in \feas$ if and only if $\mathbf{A} \cdot \bx \leq \mathbf{c}$. Without loss of generality, let $c_j = 1$ for all $j \in [m]$.

We assume that $a_{j, i} \leq \frac{1}{2}$ for all $i,j$ and that the column sparsity is bounded by $d$, meaning that for each $i$ there are at most $d$ choices of $j$ such that $a_{j,i} > 0$.

\begin{theorem}
\label{thm:sparsepips}
For $d$-sparse linear packing programs with constraint matrix $\mathbf{A} \in \reals_{\geq 0}^{m \times n}$ such that $a_{j, i} \leq \frac{1}{2}$ for all $i,j$ and unit capacities, there exist weakly $(1,0,d)$- and $(2,0,d)$-balanced prices, for fractional and integral solutions, respectively, with respect to arbitrary allocation algorithms $\ALG$. The prices can be computed by running $\ALG$ once.
\end{theorem}


We will define static prices $p_i(x_i \mid \by)$ as follows. 
Let $\bx^\ast = \ALG(v)$. For each constraint $j$, we define a per-unit price $\rho_j$ by setting $\rho_j = \sum_{i \in N: a_{j, i} > 0} v_i x_i^\ast$. Now, we define $p_i(x_i \mid \by) = \sum_j a_{j, i} x_i \rho_j$ for every quantity $x_i$ that can feasibly be added to $\by$. The claim will now follow by the two lemmas below.

\begin{lemma}
For $\feas$ being all fractional solutions, the devised pricing scheme is weakly $(1, 0, d)$-balanced with respect to $\ALG$.
\end{lemma}

\begin{proof}
Let $\feasx = \{ \bz \mid \mathbf{A}(\bx + \bz) \leq \mathbf{1} \}$. To verify Condition (a), we derive a lower bound on $\bv(\OPT(\bv, \feasx))$ as follows. Given $\bx$, let $\bz$ be defined by setting $z_i = x_i^\ast (1 - \max_{j: a_{j, i} > 0} \sum_{i'} a_{j, i'} x_{i'})$. We have $\bz \in \feasx$ because for every constraint $j$ we have 
\[
	\sum_i a_{j, i} x_i + \sum_i a_{j, i} z_i \leq \sum_i a_{j, i} x_i + \sum_i a_{j, i} x_i^\ast (1 - \sum_{i'} a_{j, i'} x_{i'}) \leq 1.
\]
Note that furthermore, by this definition, $x_{i'}^\ast (\sum_i a_{j, i} x_i) \geq x_{i'}^\ast - z_{i'}$ for all $i'$ and all $j$.

Now, it follows that
\begin{align*}
	\sum_i p_i(x_i \mid x_{[i-1]})
	&= \sum_i \sum_j a_{j, i} x_i \rho_j\\
	&= \sum_i \sum_j a_{j, i} x_i \sum_{i': a_{j, i'} > 0} v_{i'} x_{i'}^\ast\\
	&= \sum_j \sum_{i': a_{j, i'} > 0} v_{i'} x_{i'}^\ast (\sum_i a_{j, i} x_i)\\
	&\geq \sum_j \sum_{i': a_{j, i'} > 0} v_{i'} (x_{i'}^\ast - z_{i'})\\
	&\geq \sum_{i'} v_{i'} (x_{i'}^\ast - z_{i'})\\
	&\geq \bv(\ALG(\bv)) - \bv(\OPT(\bv, \feasx)).
\end{align*}

For Condition (b), we simply observe that
\begin{align*}
	\sum_i p_i(x_i' \mid x_{[i-1]}) 
	&= \sum_i \sum_j a_{j, i} x_i' \rho_j \displaybreak[0]\\
	&= \sum_j \rho_j \sum_j a_{j, i} x_i' \displaybreak[0]\\
	&\leq \sum_j \rho_j \displaybreak[0]\\
	&= \sum_j \sum_{i \in N: a_{j, i} > 0} v_i x_i^\ast \displaybreak[0]\\
	&= \sum_i \sum_{j: a_{j, i} > 0} v_i x_i^\ast \displaybreak[0]\\
	&\leq d \sum_i v_i x_i^\ast = d \cdot \bv(\ALG(\bv)). \qedhere
\end{align*}
\end{proof}

\begin{lemma}\label{lem:pip-integral}
For $\feas$ being all integral solutions, the devised pricing scheme is weakly $(2, 0, d)$-balanced with respect to $\ALG$.
\end{lemma}

\begin{proof}
To define $\feasx$, let $\mathbf{c}$ be the vector such that $c_j = 1$ if $\sum_i a_{j, i} x_i \leq \frac{1}{2}$ and $c_j = 0$ otherwise. That is, $c_j = 1$ if and only if constraint $j$ still has capacity at least $\frac{1}{2}$ after adding $\bx$. Let $\feasx$ be the set of integral solutions $\by$ that fulfill 
$\mathbf{A} \by \leq \mathbf{c}$.

Let $\bz$ be defined such that $z_i = x_i^\ast$ if $\sum_{i'} a_{j, i'} x_{i'} \leq \frac{1}{2}$ for all $j$ with $a_{j, i} > 0$ and let $z_i = 0$ otherwise. By this definition, we have $z_i \geq x_i^\ast (1 - 2 \max_{j: a_{j, i} > 0} \sum_{i'} a_{j, i'} x_{i'})$. For this reason, $x_{i'}^\ast (\sum_i a_{j, i} x_i) \geq \frac{1}{2} (x_{i'}^\ast - z_{i'})$ for all $i'$ and all $j$.

From here on, we can follow the exact same calculations as above. For Condition (a), we have 
\begin{align*}
	\sum_i p_i(x_i \mid x_{[i-1]}) 
	&= \sum_i \sum_j a_{j, i} x_i \rho_j \displaybreak[0]\\
	&= \sum_i \sum_j a_{j, i} x_i \sum_{i': a_{j, i'} > 0} v_{i'} x_{i'}^\ast \displaybreak[0]\\
	&= \sum_j \sum_{i': a_{j, i'} > 0} v_{i'} x_{i'}^\ast (\sum_i a_{j, i} x_i) \displaybreak[0]\\
	&\geq \sum_j \sum_{i': a_{j, i'} > 0} v_{i'} \frac{1}{2} (x_{i'}^\ast - z_{i'}) \displaybreak[0]\\
	&\geq \frac{1}{2} \sum_{i'} v_{i'} (x_{i'}^\ast - z_{i'}) \displaybreak[0]\\
	&\geq \frac{1}{2} \big(\bv(\ALG(\bv)) - \bv(\OPT(\bv, \feasx))\big).
\end{align*}

For Condition (b), we observe that again 
\begin{align*}
	\sum_i p_i(x_i' \mid x_{[i-1]})
	&= \sum_i \sum_j a_{j, i} x_i' \rho_j \displaybreak[0]\\
	&= \sum_j \rho_j \sum_j a_{j, i} x_i' \leq \sum_j \rho_j \displaybreak[0]\\
	&= \sum_j \sum_{i \in N: a_{j, i} > 0} v_i x_i^\ast \displaybreak[0]\\
	&= \sum_i \sum_{j: a_{j, i} > 0} v_i x_i^\ast \displaybreak[0]\\
	&\leq d \sum_i v_i x_i^\ast = d \cdot \bv(\ALG(\bv)). \qedhere
\end{align*}
\end{proof}

\paragraph{Pricing Integral Problems based on Fractional Solutions}
The way we described the pricing schemes above was to use an offline allocation algorithm $\ALG$ for the respective problem. This algorithm has to solve an NP-hard problem and as we show any approximation guarantee is preserved in the process. However, it is also possible to use the respective optimal fractional solution instead of $\ALG$. The pricing schemes defined this way then achieve the described approximation guarantees with respect to the fractional optimum. To show this result, we have to slightly extend our framework and allow $\feasx$ to contain distributions over outcome profiles that are exchange compatible.

Specifically, for the integral part of Theorem~\ref{thm:sparsepips} we would define $\feasx$ as the set of distributions such that the \emph{expected} consumption in the $j$th constraint is at most $1$ if $\sum_i a_{j, i} x_i \leq \frac{1}{2}$ and $0$ otherwise. Note that this extends the definition of $\feasx$ in the proof of Lemma~\ref{lem:pip-integral} to distributions. Following the steps above, $\bz$ is now randomized based on the optimal fractional solution $\bx^\ast$. We let $z_i = 1$ with probability $x_i^\ast$ if $\sum_{i'} a_{j, i'} x_{i'} \leq \frac{1}{2}$ for all $j$ with $a_{j, i} > 0$ and let $z_i = 0$ otherwise. In the remainder, $z_{i'}$ is replaced by its expectation.

\section{Proof of Theorem \ref{thm:submodular-matroid}}
\label{app:matroid}

We consider a matroid feasibility constraint
$\mathcal{M} = (E,\mathcal{I})$, where $E$ is a ground set of elements and $\mathcal{I} \subseteq 2^E$ is a family of feasible subsets.  Set $E$ is partitioned into subsets $E_1, \dotsc, E_n$, corresponding to agents, and for each $i$, the outcome space $X_i$ contains all subsets of $E_i$. One can think of $E_i$ as the set of elements from which agent $i$ can choose. An allocation $\bx$ is feasible if and only if $\bigcup_i x_i \in \mathcal{I}$; that is, if the chosen elements form an independent set.\footnote{Kleinberg and Weinberg~\cite{KleinbergW12} also provide bounds for a multi-dimensional matroid setting.  Their result, which has implications for revenue and welfare, applies when buyers have unit-demand preferences with independent values across elements.  Our approach is different, and provides a welfare bound for a broader class of valuations.}
Given a valuation profile $\bv$, let $\OPT(\bv) \in \mathcal{I}$ denote a social-welfare maximizing allocation.  Furthermore, given $S \in \mathcal{I}$, let $\OPT(\bv \mid  S)$ denote the set $T \in \mathcal{I}$ maximizing $\bv(T)$ such that $S \cap T = \emptyset$ and $S \cup T \in \mathcal{I}$.
Theorem \ref{thm:multidimmatroid} asserts the existence of balanced prices for matroid settings with additive, submodular, and XOS valuations.

\begin{theorem}
\label{thm:multidimmatroid}
There exist prices for multi-dimensional matroid settings in which agents have additive, submodular, or XOS valuations that are $(1,1)$-balanced with respect to $\OPT$.
\end{theorem}

As a special case, we obtain $(1,1)$-balanced prices for the single-dimensional matroid setting of \cite{KleinbergW12}, where each agent controls exactly one element of the ground set.

We prove the theorem for additive valuations. 
That is, there are non-negative numbers $v_{i, j}$ such that $v_i(x_i) = \sum_{j \in x_i} v_{i, j}$.
The proof for submodular and XOS Valuations then follows directly from
Theorem~\ref{thm:max} (composition theorem: closure under maximum).

\paragraph{Additive Valuations}

Given $S \in \mathcal{I}$, $i \in N$, and a valuation profile $\bv$, define
\begin{equation}
\label{eq.prices.matroid}
p_i^\bv(x_i \mid \by) = \begin{cases}
\bv(\OPT(\bv \mid \bigcup_j y_j)) - \bv(\OPT(\bv \mid (\bigcup_j y_j \cup x_i))) & \text{, if $(\bigcup_j y_j \cup x_i) \in \mathcal{I}$} \\
\infty & \text{, otherwise.}
\end{cases}
\end{equation}

The following lemma shows that the prices in Equation \ref{eq.prices.matroid} are monotonically increasing.
We then prove in Theorem \ref{thm.matroid.balanced} that these prices are $(1,1)$-balanced.

\begin{lemma}
\label{lemm:mono}
Consider the pricing rule $p^\bv$. For an arbitrary profile $\bv$, feasible outcome profiles $\by, \by' \in \feas$ such that $y_j \subseteq y'_j$ for all $j$, agent $i$, and allocation $x_i$,
\[
	p^\bv_i(x_i \mid \by) \leq p_i^\bv (x_i \mid \by').
\]
\end{lemma}
\begin{proof}
Let $S = \bigcup_j y_j$, $T = \bigcup_j y_j'$. If $S \cup x_i \not\in \mathcal{I}$, then $T \cup x_i \not\in \mathcal{I}$, and so $p^\bv_i(x_i \mid \by) = p^\bv_i(x_i \mid \by') = \infty$. Otherwise,
\begin{align*}
	p^\bv_i(x_i \mid \by)
	&= \bv(\OPT(\bv \mid S)) - \bv(\OPT(\bv \mid (S \cup x_i))) \\
	&\leq \bv(\OPT(\bv \mid T)) - \bv(\OPT(\bv \mid (T \cup x_i))) \\
	&\leq p^\bv_i(x_i \mid \by'),
\end{align*}
where the equation follows by the definition of $p^\bv_i(x_i \mid \by)$, and the inequality follows from the submodularity of the function $f(U) = \bv(\OPT(\bv \mid U))$ (cf.~\cite[Lemma 3]{KlWe12}).
\end{proof}

\begin{theorem}
\label{thm.matroid.balanced}
The pricing rule defined in Equation \ref{eq.prices.matroid} is $(1,1)$-balanced with respect to the optimal allocation rule $\OPT$. 
\end{theorem}

\begin{proof}
We define $\feasx$ as the set of all outcome profiles $\by \in E_1 \times \ldots \times E_n$ such $\left(\bigcup_i y_i\right) \cap \left(\bigcup_i x_i\right) = \emptyset$ and $\left(\bigcup_i y_i\right) \cup \left(\bigcup_i x_i\right) \in \mathcal{I}$. In other words, we consider the matroid contracted by the set $\bigcup_i x_i$.

We first show that Condition \eqref{conda} holds for $\alpha = 1$.
For every agent $i \in N$ and every $\bx \in \feas$ by a telescoping-sum argument,
\begin{align*}
\sum_i p_i^\bv(x_i \mid \bx_{[i-1]}) & = \sum_i \left[ \bv(\OPT(\bv \mid \bigcup_{j=1}^{i-1} x_j)) - \bv(\OPT(\bv \mid \bigcup_{j=1}^i x_j)) \right] \\
& = \bv(\OPT(\bv)) - \bv(\OPT(\bv \mid \bigcup_j x_j)) \\
& = \bv(\OPT(\bv)) - \bv(\OPT(\bv, \feasx)).
\end{align*}

We next show that the second condition holds for $\beta = 1$.  Consider some arbitrary $\bx \in \feas$, $\bx' \in \feasx$.  Let $S = \bigcup_i x_i$. Note that $\OPT(\bv \mid S)$ is precisely a maximum-weight basis of the matroid contracted with $S$.  By the generalized Rota exchange theorem~\cite[Lemma 2.7]{LeeSV10}, for each $i$ there exists some set $R_i \subseteq \OPT(\bv \mid S)$ (which may not be contained in $E_i$) such that each element of $\OPT(\bv \mid S)$ appears in exactly one $R_i$, and $(\OPT(\bv \mid S) \setminus R_i) \cup x'_i$ is an independent set in the matroid after contracting $S$.  We therefore have
\begin{align*}
	\sum_{i \in N} p_i^\bv(x'_i \mid \bx_{[i-1]})
	&\leq \sum_{i \in N} p_i^\bv(x'_i \mid \bx)\\
	&= \sum_{i \in N} \bigg[\bv(\OPT(\bv \mid S)) - \bv(\OPT(\bv \mid S \cup x'_i))\bigg]\\
	&\leq \sum_{i \in N} \bigg[\bv(\OPT(\bv \mid S)) - \bv(\OPT(\bv \mid S) \setminus  R_i)\bigg]\\
	&= \sum_{i \in N} \bv(R_i)\\
	&= \bv(\OPT(\bv \mid \bx)),
\end{align*}
where the first inequality follows from the monotonicity of the prices (Lemma \ref{lemm:mono}), and the second inequality follows by observing that $(\OPT(\bv) \setminus R_i)\cup x'_i$ is feasible in the contracted matroid (see the Rota exchange argument above), and thus $\bv(\OPT(\bv) \setminus R_i) \leq \bv(\OPT(\bv \mid S \cup x'_i))$.  The final equality follows by additivity.  We conclude that the suggested prices are $(1,1)$-balanced.
\end{proof}




\paragraph{Computational Aspects}
Theorem \ref{thm:multidimmatroid} establishes the existence of a price-based $2$-approximate prophet inequality for additive, submodular, and XOS preferences. For additive valuations the construction is polytime, as the greedy algorithm is optimal. For submodular and XOS valuations computing the optimal allocation is NP-hard.

We claim that the result for submodular valuations can turn into a computational one, by basing the prices on an approximately optimal allocation. Namely, let $\GRD$ be the algorithm that allocates items greedily by value, always choosing the item that locally increases $\bv(\bx)$ the most subject to the matroid constraint. We show that $\GRD$ is consistent (see Definition \ref{def:consistent}) in this setting.

\begin{lemma}
The greedy algorithm $\GRD$ is consistent for XOS valuations over matroid feasibility constraints.
\end{lemma}

\begin{proof}
Fix an XOS valuation profile $\bv$, and let $\btv$ be a supporting valuation profile for the allocation $\GRD(\bv)$. Note that as the supporting valuation profile is additive, $\GRD$ returns the optimal solution. As $\GRD(\bv)$ is another feasible solution, we have $\btv(\GRD(\btv)) \geq \btv(\GRD(\bv))$.
\end{proof}

Theorem~\ref{thm:max} (closure under maximum) now implies that the prices given in~\eqref{eq.prices.matroid} for an additive valuation supporting $\GRD(\bv)$ are $(1,1)$-balanced with respect to the greedy allocation.
We can compute these prices in polynomial time by simulating $\GRD$ and then determining the supporting additive valuation.
Since $\GRD$ is a 2-approximation to $\OPT$ for submodular valuations, Theorem~\ref{thm:extension} then implies that we can compute prices that yield a $4$-approximation to the optimal expected welfare, less an additive sampling error.

\section{A Smooth Mechanism Without Good Posted Prices}
\label{app:smoothness-is-weaker}

In this appendix we show that there are allocation problems that admit constant-factor smooth mechanisms, but for which no posted-price mechanism can guarantee more than a linear fraction of the optimal social welfare.

\begin{proposition}
There exists a downward-closed welfare maximization problem that admits a $(1, 0)$-smooth mechanism, but for which any posted price mechanism has approximation factor $\Omega(n)$.
\end{proposition}
\begin{proof}
Let $k$ be a positive integer to be fixed later.  In the allocation problem we consider, $X_i = [k]^n \cup \{\emptyset\}$ for each $i \in [n]$.  That is, each agent is allocated either $\emptyset$ or a sequence of $n$ integers between $1$ and $k$.  For $x_i \in X_i \backslash \{\emptyset\}$, write $x_i = (x_{i1}, \dotsc, x_{in})$, where each $x_{ij} \in [k]$.  The set of feasible allocations is $\feas = \{ \allocs \colon (\alloc_i \neq \alloc_j) \implies ((\alloc_i = \emptyset) \vee (\alloc_j = \emptyset)) \}$.  That is, all agents who receive a non-empty allocation must receive the same allocation.

Agents have the following valuations.  Each agent $i$ has some desired value $z_i \in [k]$.  The value of an allocation is $v_i(x_i) = 1$ if $x_{ii} = z_i$, and $v_i(x_i) = 0$ otherwise.  Let $\disti$ be the distribution over such valuations in which $z_i$ is chosen uniformly from $[k]$.

For this feasibility constraint and space of valuations, consider the mechanism that returns the welfare-optimal allocation and charges payments of $0$.  This mechanism simultaneously satisfies all desires by allocating $(z_1, \dotsc, z_n)$ to every agent, where $z_i$ is the desired value reported by agent $i$.  This mechanism is $(1,0)$-smooth, with truth-telling being the required deviation.

On the other hand, consider any posted-price mechanism.  Whichever is the first agent to obtain a non-$\emptyset$ outcome, say agent $i$ purchasing allocation $\alloci$, each subsequent agent $j$ can obtain positive value only if $z_j = \alloc_{ij}$, which occurs with probability $1/k$.  We therefore have that the expected welfare of any posted price mechanism is at most $1 + \frac{n-1}{k}$.  Taking $k = n$ and noting that the optimal welfare is $n$, we conclude that no posted price mechanism can obtain more than an $\Omega(n)$ approximation to the optimal welfare.
\end{proof}

\section{Proof of Theorem \ref{thm:smoothness-to-prices}}
\label{app:smooth-prices}


To prove Theorem~\ref{thm:smoothness-to-prices}, we will first introduce the notion of weak outcome smoothness, then show how to construct the prices, and finally establish the two conditions necessary for \AB separately.

\subsection{Weak Outcome Smoothness and Its Properties}
\label{app:outcome-smooth}

We begin by defining a more general notion of outcome smoothness, corresponding to the notion of weak smoothness.

\begin{definition}\label{def:weak-outcome-smooth}
A mechanism is weakly \emph{$(\lambda,\mu_1,\mu_2)$-outcome smooth} for $\lambda,\mu_1,\mu_2 \geq 0$ if for all valuation profiles $\vals \in V$ there exists an outcome $\allocsprime(\vals) \in \feas$ such that for all bid profiles $\bids \in B$,
\[
	\sum_{i \in N} \left( v_i(\allociprime) - \inf_{\bidiprime:\; \mechallocrule_i(\bidiprime,\bidsmi) \succeq \allociprime} P_i(\bidiprime,\bidsmi) \right) \geq \lambda \cdot \vals(\OPT(\vals)) - \mu_1 \cdot \sum_{i \in N} P_i(\bids) - \mu_2 \cdot \bids(\mechallocrule(\bids)).
\]
\end{definition}

The following observation will facilitate our discussion in what follows. Under a mild scale invariance assumption, we can without loss of generality consider $(\lambda,0,\mu)$-outcome smooth mechanisms.  


\begin{proposition}
\label{prop:outcome-smoothness-remove-revenue}
Consider a $(\lambda,\mu_1,\mu_2)$-outcome smooth mechanism with allocation rule $\mechallocrule$ and payment rule $P$ such that for every $\delta > 0$, we have $\allocsprime(\bv / \delta) = \allocsprime(\bv)$. For every player $i \in N$ and all bid profiles $\bids \in B$ define
\[
	\tilde{P}_i(\bids) = \min\left\{ \frac{\bids(\mechallocrule(\bids))}{\sum_{i \in N} P_i(\bids)}, 1 \right\} \cdot P_i(\bids).
\]
Then the mechanism with allocation rule $\mechallocrule$ and payment rule $\tilde{P}$ is $(\lambda,0,\mu_1+\mu_2)$-outcome smooth.
\end{proposition}
\begin{proof}
Given a valuation profile $\vals$, we claim that the re-defined mechanism is $(\lambda,0,\mu_1+\mu_2)$-outcome smooth with respect to the same outcome $\allocsprime(\vals)$.

Fix a bid profile $\bids \in B$ and let $\delta = \min\left\{ \frac{\bids(\mechallocrule(\bids))}{\sum_{i \in N} P_i(\bids)}, 1 \right\}$. Note that $\delta \leq 1$.
Using outcome smoothness of the original mechanism for valuation profile $\vals/\delta$ and corresponding outcome $\allocsprime(\vals/\delta) = \allocsprime(\vals)$ we obtain
\[
	\sum_{i \in N} \left( \frac{1}{\delta} \cdot v_i(\allociprime) - \inf_{\bidiprime:\; \mechallocrule(\bidiprime,\bidsmi) \succeq \allociprime} P_i(\bidiprime,\bidsmi) \right) \geq \lambda \cdot \frac{1}{\delta} \cdot \vals(\OPT(\vals)) - \mu_1 \cdot \sum_{i \in N} P_i(\bids) - \mu_2 \cdot \bids(\mechallocrule(\bids)).
\]
Multiplying by $\delta$, we have
\[
	\sum_{i \in N} \left( v_i(\allociprime) - \delta \cdot \inf_{\bidiprime:\; \mechallocrule(\bidiprime,\bidsmi) \succeq \allociprime} P_i(\bidiprime,\bidsmi) \right) \geq \lambda \cdot \vals(\OPT(\vals)) - \mu_1 \cdot \delta \cdot \sum_{i \in N} P_i(\bids) - \mu_2 \cdot \delta \cdot \bids(\mechallocrule(\bids)).
\]
From the definition of $\tilde{P}$ and $\delta$ we know that $\delta \cdot P_i(\bidiprime,\bidsmi) = \tilde{P}_i(\bidiprime,\bidsmi)$, $\delta \cdot \sum_{i \in N} P_i(\bids) \leq \bids(\mechallocrule(\bids))$, and  $\delta \cdot \bids(\mechallocrule(\bids)) \leq \bids(\mechallocrule(\bids))$.
This yields
\[
	\sum_{i \in N} \left( v_i(\allociprime) - \inf_{\bidiprime:\; \mechallocrule(\bidiprime,\bidsmi) \succeq \alloci'} \tilde{P}_i(\bidiprime,\bidsmi) \right) \geq \lambda \cdot \vals(\OPT(\vals)) - (\mu_1 + \mu_2) \cdot \bids(\mechallocrule(\bids)). \qedhere
\]
\end{proof}

Another observation is that the function $\bx'$ of a weakly $(\lambda, \mu_1, \mu_2)$-outcome smooth mechanism is indeed a $\lambda$-approximation of social welfare, which also implies that its range $\lambda$-approximates the space of all outcome profiles.

\begin{lemma}\label{lem:onto}
Suppose $Y$ is the range of $\allocs'$.  Then
\[
	\max_{\mathbf{y} \in Y} \sum_{i \in N} \vali(y_i) \geq \sum_{i \in N} \vali(\allocs'(\vals)) \geq \lambda \cdot \vals(\OPT(\vals)).
\]
\end{lemma}
\begin{proof}
We use outcome smoothness setting the valuation profile to $\vals$ and the bid profile to $0$. Then for $\mathbf{y} = \allocs'(\vals)$,
\[
	\sum_{i \in N} \left( \vali(y_i) - \inf_{\bidiprime:\; \mechallocrule(\bidiprime,0) \succeq y_i} P_i(\bidiprime,0) \right) \geq \lambda \cdot \vals(\OPT(\vals))\enspace.
\]
As we assume that payments are never negative, this implies
\[
	\sum_{i \in N} \vali(y_i) \geq \lambda \cdot \vals(\OPT(\vals))\enspace. \qedhere
\]
\end{proof}

\subsection{Constructing the Prices}

We first have to define $p_i(x_i \mid \bz)$ for all $i \in N$, $\bz \in \feas$, $x_i \in X_i$. To this end, we create two additional copies of each agent, implying that we have $3n$ agents overall. The different incarnations of agent $i$, denoted by $i$, $i+n$, and $i+2n$ correspond to different roles when setting the prices. They particularly ensure that agent $i$ competes against itself.

In more detail, for $i \in [n]$, the roles will be as follows in defining $p_i(x_i \mid \bz)$. Agent $i$ is used to represent $v_i$, agent $i + n$ is used to represent $z_i$ and agent $i + 2n$ tries to buy outcome $x_i$ in the outcome-smooth mechanism. We will frequently map outcome profiles $\bz \in \feas$ from the original space of $n$ agents to the agents $n+1, \ldots, 2n$. This operation we denote by $\overline{\bz}$. That is $\overline{z}_i = z_{i - n}$ for $i \in \{n+1, \ldots, 2n\}$ and $\overline{z}_i = \emptyset$ otherwise.

We will adopt the more general notion of weak outcome smoothness, from Appendix~\ref{app:outcome-smooth}. We assume that outcome smoothness holds in every subinstance $\feas'/\bz$ of the extended outcome space $\feas'$. To avoid any possible confusion of the numbering of the agents, we denote by $\mathcal{G}_{\bz}$, the outcome space $\feas'/\bz$ padded with $\emptyset$ wherever $z_i \neq \emptyset$. That is, $\mathcal{G}_{\bz}$ is isomorphic to $\feas'/\bz$ but uses the exact same indices as $\feas'$. We will denote the allocation rule of the outcome smooth mechanism in this space by $\mechallocrule(\,\cdot\,\mid \bz)$ and its payment rule by $P(\,\cdot\,\mid \bz)$.

Outcome smoothness now ensures that for all $\bz \in \feas'$ and valuation profiles $\vals \in V$ there exists an outcome $\allocsprime(\vals \mid \bz) \in \mathcal{G}_{\bz}$ such that for all bid profiles $\bids \in B$,
\[
	\sum_{i = 1}^{3n} \left( v_i(\allociprime) - \inf_{\bidiprime:\; \mechallocrule_i((\bidiprime,\bidsmi) \mid \bz) \succeq \allociprime} P_i((\bidiprime,\bidsmi) \mid \bz) \right) \geq \lambda \cdot \vals(\OPT(\vals, \mathcal{G}_{\bz})) - \mu \cdot \bids(\mechallocrule(\bids \mid \bz)).
\]

Given these definitions, we define prices based on the payments in the outcome-smooth mechanism by
\[
	p_i^\vals(x_i \mid \bz) = \inf_{b'_{i+2n}: \; \mechallocrule_i\left((b'_{i+2n},\vals_{-(i+2n)}) \growingmid \overline{\bz}\right) \succeq x_i} P_i\left((b'_{i+2n},\vals_{-(i+2n)}) \growingmid \overline{\bz}\right)\enspace.
\]
The meaning is as follows: Agents $n+1, \ldots, 2n$ incorporate the allocation $\bz$ and agents $1, \ldots, n$ represent the part of $\bv(\OPT(\bv))$ that is not yet taken away by $\bz$. Now, another copy of agent $i$ is added to the system and competes with agents $1, \ldots, n$ for the outcome.

\subsection{Showing Balancedness}

Let $\ALG$ denote the algorithm that returns $\OPT(\bv)$ with probability $\lambda$. We will show that the constructed prices are $(\lambda, \mu/\lambda)$-balanced with respect to $\ALG$ if they are monotone. To this end, we define $\feasx$ as the range of $\bx'(\,\cdot\, \mid \bz)$ on agents $2n + 1, \ldots, 3n$. Formally,
\[
\feasx = \{ \by \in X_1 \times \ldots \times X_n \mid \exists \bv: \bx'( \bv \mid \overline{\bx})  = \overline{\overline{\by}} \},
\]
where $\overline{\overline{\by}}$ is the outcome profile on $3n$ agents that sets $\overline{\overline{y}}_{2n+i} = y_i$ for $i \in [n]$ and $\overline{\overline{y}}_i = \emptyset$ for $i \leq 2n$.

\begin{lemma}\label{lem:outcome-beta}
Suppose $(\mechallocrule(\, \cdot \mid \bz), P(\, \cdot \mid \bz))$ is $(\lambda,0,\mu)$-outcome smooth on every $\mathcal{G}_{\bz}$ such that $\allocsprime(\vals \mid \bz) = \allocsprime(\epsilon \cdot \vals \mid \bz)$ for every $\epsilon > 0$, $p_i^{\vals}$ is monotonically increasing, then $p_i^{\vals}$ satisfies Condition (\ref{condb}) of $(\alpha, \beta)$-balancedness with $\beta = \mu/\lambda$, for the choice of $(\feasx)_{\bx \in X}$ described above.
\end{lemma}

\begin{proof}
Consider $\allocs$ and $\mathbf{y} \in \feasx$. By definition of $\feasx$, we can find some $\vals'$ such that $\allocs'(\vals') = \overline{\overline{\mathbf{y}}}$. Recall that, in allocation $\overline{\overline{\mathbf{y}}}$, agents $1$ through $2n$ obtain the empty allocation, and agents $2n+1$ through $3n$ receive allocation profile $\by$.

By assumption for all $\epsilon > 0$ we have $\allocsprime(\epsilon \cdot \valsprime) = \overline{\overline{\mathbf{y}}}$ . Outcome smoothness for problem subinstances, with valuation profile $\epsilon \cdot \vals'$ and bid profile $\vals$ implies
\[
	\sum_{i = 1}^{3n} \left( \epsilon \cdot \valiprime(\overline{\overline{y}}_i) - \inf_{\bidiprime:\; \mechallocrule_i((\bidiprime,\valsmi) \mid \bx) \succeq \overline{\overline{y}}_i} P_i((\bidiprime,\valsmi) \mid \allocs) \right) \geq \lambda \cdot \epsilon \cdot \valsprime(\OPT(\epsilon \cdot \valsprime, \mathcal{G}_{\bx})) - \mu \cdot \vals(\mechallocrule(\vals \mid \allocs))\enspace.
\]
As payments $P_i$ are never negative, this implies
\[
	\sum_{i = 2n+1}^{3n} \left( - \inf_{\bidiprime:\; \mechallocrule_i((\bidiprime,\valsmi) \mid \bx) \succeq \overline{\overline{y}}_i} P_i((\bidiprime,\valsmi) \mid \allocs) \right) \geq \lambda \cdot \epsilon \cdot \valsprime(\OPT(\epsilon \cdot \valsprime, \mathcal{G}_{\bx})) - \mu \cdot \vals(\mechallocrule(\vals \mid \allocs)) - \sum_{i = 1}^{3n} \epsilon \cdot \valiprime(\overline{\overline{y}}_i) \enspace.
\]

This implies
\begin{align*}
	\sum_{i = 1}^n p_i^{\vals}(y_i \mid \allocs) & = \sum_{i = 2n+1}^{3n} \inf_{\bidiprime:\; \mechallocrule_i((\bidiprime,\valsmi) \mid \bx) \succeq \overline{\overline{y}}_i} P_i((\bidiprime,\valsmi) \mid \allocs) \\
	& \leq \mu \cdot \vals(\mechallocrule(\vals \mid \allocs)) - \epsilon \left( \sum_{i = 1}^{3n} \valiprime(\overline{\overline{y}}_i) - \lambda \cdot \valsprime(\OPT(\epsilon \cdot\valsprime, \mathcal{G}_{\bx})) \right).
\end{align*}
As this holds for all $\epsilon > 0$, we also have
\[
	\sum_{i = 1}^n p_i^{\vals}(y_i \mid \allocs) \leq \mu \cdot \vals(\mechallocrule(\vals \mid \allocs)) \leq \mu \cdot \vals(\OPT((v_1, \ldots, v_n, 0, \ldots, 0), \mathcal{G}_{\bx})) \leq \frac{\mu}{\lambda} \vals(\OPT(\vals, \feasx)) \enspace,
\]
where the last step uses Lemma~\ref{lem:onto}.
\end{proof}

\begin{lemma}\label{thm:outcome-alpha}
Suppose $(\mechallocrule(\, \cdot \mid \bz), P(\, \cdot \mid \bz))$ is $(\lambda,0,\mu)$-outcome smooth on every $\mathcal{G}_{\bz}$ such that $\allocsprime(\vals \mid \bz) = \allocsprime(\epsilon \cdot \vals \mid \bz)$ for every $\epsilon > 0$, $\mechallocrule$ is the declared welfare maximizer, and $P$ is the first-price payment rule.  Then $p_i^{\vals}$ satisfies Condition (\ref{conda}) of $(\alpha, \beta)$-balancedness with $\alpha = \lambda$, for the choice of $\ALG$ and $(\feasx)_{\bx \in X}$ described above.
\end{lemma}

\begin{proof}
Consider an arbitrary player $i$ and arbitrary outcomes $\allocs \in \feas$. To bound $p_i(x_i \mid \allocs_{[i-1]})$, we consider player $2n+i$ in the outcome-smooth mechanism. Let $b_{2n+i}$ be a bid such that $\mechallocrule_{2n+i}(b_{2n+i}, \vals_{-(2n+i)} \mid \overline{\allocs}_{[n+i-1]}) \succeq x_i$. We show that for any such $b_{2n+i}$, we have
\[
P_i(b_{2n+i}, \vals_{-(2n+i)} \mid \overline{\allocs}_{[n+i-1]}) \geq \bv(\OPT(\bv, \mathcal{G}_{\overline{\allocs}_{[n+i-1]}})) - \bv(\OPT(\bv, \mathcal{G}_{\overline{\allocs}_{[n+i]}}))\enspace.
\]
As this holds for all $b_{2n+i}$, this gives also a lower bound on the infimum.

In the following, we keep the allocation $\overline{\allocs}_{[n+i-1]} = (\emptyset, \ldots, \emptyset, x_1, \ldots, x_{i-1}, \emptyset, \ldots \emptyset)$ fixed and compare the two possible feasible solutions $\mathbf{a} := \mechallocrule(b_{2n+i}, \vals^{(i-1)}_{-(2n+i)} \mid \overline{\allocs}_{[n+i-1]})$ and $\mathbf{q} := \OPT(\bv, \mathcal{G}_{\overline{\allocs}_{[n+i-1]}})$.

As $\mathbf{a}$ maximizes $(b_{2n+i}, v_{-(2n+i)})$ when keeping $\overline{\allocs}_{[n+i-1]}$ fixed, we have
\[
b_{2n+i}(a_{2n+i}) + \sum_{\substack{j = 1\\ j \neq 2n + i}}^{3n} v_j(a_j) \geq b_{2n+i}(q_{2n+i}) + \sum_{\substack{j = 1\\ j \neq n + i}}^{3n} v_j(q_j) \enspace.
\]
As $v_j = 0$ for all $j > n$, this is equivalent to
\[
b_{2n+i}(a_{2n+i}) \geq b_{2n+i}(q_{2n+i}) + \sum_{j = 1}^n v_j(q_j) - \sum_{j = 1}^n v_j(a_j)\enspace.
\]


Now, observe that if we replace $a_{2n+i}$ by $x_i$ in allocation $\mathbf{a}$, then the modified vector $\mathbf{a}'$ is still feasible in the space that keeps $\overline{\allocs}_{[n+i-1]}$ fixed. By symmetry of players $2n+i$ and $n+i$, this also means that $(a_1, \ldots, a_n, \emptyset, \ldots, \emptyset) \in \mathcal{G}_{\overline{\allocs}_{[n+i]}}$. This implies $\sum_{j = 1}^n v_j(a_j) \leq \bv(\OPT(\bv, \mathcal{G}_{\overline{\allocs}_{[n+i]}}))$.

Overall, we get
\[
p_i^{\vals}(x_i \mid \allocs_{[i-1]}) \geq \bv(\OPT(\bv, \mathcal{G}_{\overline{\allocs}_{[n+i-1]}})) - \bv(\OPT(\bv, \mathcal{G}_{\overline{\allocs}_{[n+i]}})) \enspace.
\]
Summing up the prices for all players $i$ and using a telescoping sum, this implies a lower bound of
\begin{align*}
\sum_{i = 1}^n p_i^{\vals}(x_i \mid \allocs_{[i-1]}) & \geq \bv(\OPT(\bv, \mathcal{G}_{\overline{\allocs}_{\emptyset}})) - \bv(\OPT(\bv, \mathcal{G}_{\overline{\allocs}_{[n]}})) \\
& = \bv(\OPT(\bv)) - \bv(\OPT(\bv, \mathcal{G}_{\overline{\allocs}})) \enspace.
\end{align*}
Finally, we have $\bv(\OPT(\bv)) = \frac{1}{\lambda} \bv(\ALG(\bv))$ and $\bv(\OPT(\bv, \mathcal{G}_{\overline{\allocs}})) \leq \frac{1}{\lambda} \vals(\OPT(\vals, \feasx))$ by Lemma~\ref{lem:onto}. So, in combination
\[
\sum_{i = 1}^n p_i^{\vals}(x_i \mid \allocs_{[i-1]}) \geq \frac{1}{\lambda} \left(\vals(\ALG(\vals)) - \vals(\OPT(\vals, \feasx))\right) \enspace. \qedhere
\]
\end{proof}



\section{Proofs of Theorem \ref{thm:greedy} and Theorem \ref{thm:opt}}
\label{app:single-parameter}

In this appendix we describe how to obtain balanced prices from smooth mechanisms in binary, single-parameter settings. In these settings players can either ``win'' or ''lose'', and have a value $v_i \in \mathbb{R}_{\ge 0}$ for winning.  Feasible solutions $\bx \in \feas \subseteq \{0,1\}^n$ are subsets of players that can win simultaneously. For ease of notation we identify the vectors $\bx \in \feas$ with the subset of players $i \in N$ for which $x_i = 1$. This lets us write $i \in \bx$ if $x_i = 1$ and $i \not\in \bx$ otherwise.


\subsection{Permeable Allocation Rules}
\label{app:permeability}

We begin by defining the \emph{permeability} of an allocation rule $\mechallocrule$, and by showing that $(\lambda,\mu)$-smoothness implies a bound on permeability.

An algorithm $\mechallocrule$ for a binary, single-parameter problem is \emph{monotone} if for every player $i \in N$, any two bids $b'_i \geq b_i$, and any bid vector $\bidsmi$,
$$\mechallocrule_i(b_i,\bidsmi) = 1 \quad \Rightarrow \quad \mechallocrule_i(b'_i,\bidsmi) = 1.$$

For monotone allocation rules the \emph{critical value} for player $i$ is the smallest bid that ensures that player $i$ wins against bids $\bidsmi$. That is, $\tau^\mechallocrule_i(\bidsmi) = \inf \{b_i \mid \mechallocrule_i(b_i,\bidsmi) = 1\}$.

\begin{definition}[\citet{DuettingK15}, see also \cite{LucierB10,SyrgkanisT13,HartlineHT14}]
A monotone allocation algorithm $\mechallocrule$ for a binary, single parameter problem $\mathcal{F}$ is $\gamma$-permeable if $\gamma \geq 1$ is the smallest multiplier such that for all bid vectors $\bids$ and all feasible allocations $\allocs \in \feas$ it holds that
\begin{align}
	\sum_{i \in N: \; x_i = 1} \tau^\mechallocrule_i(\bidsmi) \leq \gamma \cdot \bids(\mechallocrule(\bids)). \label{eq:perm-def}
\end{align}
\end{definition}

\begin{theorem}[\citet{DuettingK15}]\label{thm:permeable}
Suppose that the first-price mechanism $\mathcal{M}$ based on allocation $\mechallocrule$ is $(\lambda,\mu)$-smooth, then $\mechallocrule$ is \emph{$\gamma$-permeable} with $\gamma \leq (\mu+1)/\lambda$.
\end{theorem}

\begin{proof}
Given a bid vector $\bids$, we have to show that $\sum_{i: \; x_i = 1} \tau^\mechallocrule_i(\bidsmi) \leq \frac{\mu+1}{\lambda} \cdot \bids(\mechallocrule(\bids))$. To this end, consider fixed $\bx \in \feas$ and $\bids$. Let $\epsilon > 0$ and let $\vals$ be defined by $v_i = \max\{ b_i, \tau^\mechallocrule_i(\bidsmi)\}$. By smoothness of $\mathcal{M}$, there are $\bidiprime$ such that
\[
	\sum_{i \in N} u_i(\bidiprime,\bidsmi) \geq \lambda \cdot \vals(\OPT(\vals)) - \mu \cdot \sum_{i \in N} P_i(\bids).
\]
Observe that if $i \not\in \mechallocrule(\bids)$, then $u_i(\bidiprime,\bidsmi) \leq 0$ because $\tau^\mechallocrule_i(\bidsmi) > b_i$ and this means that $i \not\in \mechallocrule(\bidiprime,\bidsmi)$ unless $\bidiprime > v_i$. None of these choices results in positive utility. Furthermore, for $i \in \mechallocrule(\bids)$, we have $u_i(\bidiprime,\bidsmi) \leq v_i = b_i$. Therefore $\sum_{i \in N} u_i(\bidiprime,\bidsmi) \leq \bids(\mechallocrule(\bids))$.

Next, we can lower-bound $\vals(\OPT(\vals))$ by the value of the feasible solution $\bx$, which gives us $\vals(\OPT(\vals)) \geq \sum_{i: \; x_i = 1} v_i \geq \sum_{i: \; x_i = 1} (\tau^\mechallocrule_i(\bidsmi) - \epsilon) \geq \sum_{i: \; x_i = 1} \tau^\mechallocrule_i(\bidsmi) - n \epsilon$.

Finally, as $\mathcal{M}$ is first-price, we also have $\sum_{i \in N} P_i(\bids) = \bids(\mechallocrule(\bids))$.

In combination this yields
\[
\bids(\mechallocrule(\bids)) \geq \lambda \cdot \left( \sum_{i: \; x_i = 1} \tau^\mechallocrule_i(\bidsmi) - n \epsilon \right) - \mu \cdot \bids(\mechallocrule(\bids)),
\]
which implies
\[
\sum_{i: \; x_i = 1} \tau^\mechallocrule_i(\bidsmi) \leq \frac{\mu+1}{\lambda} \cdot \bids(\mechallocrule(\bids)) + n \epsilon.
\]
As this holds for all $\epsilon > 0$, this shows the claim.
\end{proof}

Applying Theorem \ref{thm:permeable} to each problem in a collection of problems $\Pi$, we see that if a mechanism $\mathcal{M}$ is $(\lambda,\mu)$-smooth for $\Pi$ then it is $(\mu+1)/\lambda$-permeable for $\Pi$.

\begin{remark}
While the definition of permeability requires $\gamma$ to be the smallest multiplier for which inequality (\ref{eq:perm-def}) is satisfied, all our results can be derived from any upper bound on this multiplier at the cost of slightly worse guarantees.
\end{remark}


\subsection{Proof of Theorem \ref{thm:greedy}}
\label{app:greedy}


In this subsection we prove Theorem \ref{thm:greedy}, which shows that $(\lambda,\mu)$-smoothness of the greedy allocation rule for a subinstance-closed closed collection of binary, single-parameter problems $\Pi$ implies the existence of a weakly $((\mu+1)/\lambda,0,(\mu+1)/\lambda)$-balanced pricing rule. By Theorem \ref{thm:permeable}, in order to show this result, it suffices to show the following theorem.

\begin{theorem}\label{thm:greedy-perm}
Let $\ALG$ be any allocation rule. Suppose that the greedy allocation rule $\GRD$ is $\gamma$-permeable for a subinstance-closed collection of binary, single-parameter feasibility problems $\Pi$. Then for every $\bv \in V$ there exists a pricing rule that is weakly $(\gamma,0,\gamma)$-balanced with respect to $\ALG$ and the canonical exchange-feasible sets $(\feas_\bx)_{\bx \in X}$.
\end{theorem}


We first describe the pricing rule that achieves this result. Afterwards, we show that this pricing rule has the desirable properties. 

\subsubsection{Construction of the Prices}
\label{app:greedy-prices}

We set the price $p_i(z_i \mid \allocs)$ for player $i \in N$ and outcome $z_i \in \{0,1\}$ for arbitrary but fixed valuations $\vals$ and allocation $\allocs \in \F$ through Algorithm \ref{alg:greedy-prices}.
For this let $\GRD(\vals \mid \bx) \in \feasx$ denote the allocation that results if we go through the players in order of non-increasing value but only add a player if he is not in $\bx$ and feasible together with $\bx$ and the previously accepted players.

We generally set $p_i(0 \mid \bx) = 0$. That is, the price for losing is always zero. To determine the price $p_i(1 \mid \bx)$ for winning we first compute a sequence of reference allocations $\br^{(0)} \geq \dots \geq \br^{(n)}$ and a sequence of reference valuations $\vals^{(0)} \geq \dots \geq \vals^{(n)}$. We then set $p_i(1 \mid \bx) = v_i$ if $i \in \br^{(n)}$, i.e., the price of player $i$ is that player's valuation if he is part of the final reference allocation. Otherwise, we set $p_i(1 \mid \bx) = \inf \{v'_i: \; i \in \GRD(v'_i,\vals^{(n)}_{-i} \mid \allocs)\}$, i.e., we set the price to the player's critical value against the players in the final reference allocation.

While we need to define prices $p_i(z_i \mid \allocs)$ for any possible allocation $\bx \in \feas$, the prices that player $i$ will actually see are the ones where $\bx$ is set to the purchase decisions of the players $j = 1, \dots, i-1$ that precede player $i$ in the ordering. Note that in this case $x_i = x_{i+1} = x_n = 0$ and therefore $\br^{(n)} = \dots = \br^{(i-1)}$ and $\vals^{(n)} = \dots = \vals^{(i-1)}$. We use the shorthand $\tau_i^{\GRD}(\vals^{(i-1)}_{-i} \mid \allocs_{[i-1]}) := \inf \{v'_i: \; i \in \GRD(v'_i,\vals^{(i-1)}_{-i} \mid \allocs_{[i-1]})\}$.

\begin{algorithm}
\DontPrintSemicolon
\KwIn{$z_i \in \{0,1\}$, $\bv$, $\bx \in \feas$}
\KwOut{$p_i(z_i \mid \allocs)$}
\If{$z_i = 0$}{%
// In this case the price is simply zero\;
$p_i(z_i \mid \allocs) = 0$%
} 
\Else{%
	// First determine reference allocation and valuations\;
	$\br^{(0)} \gets \ALG(\bv)$,
	$v^{(0)}_k \gets v_k$ if $k \in \br^{(0)}$ and $v^{(0)}_k \gets 0$ otherwise\;
	\For{$j \gets 1$ {\bf to} $n$}{
		$\br^{(j)} \gets \GRD(\vals^{(j-1)} \mid \bx_{[j]})$,
		$v^{(j)}_k \gets v_k^{(j-1)} = v_k$ if $k \in \br^{(j)}$ and
		$v^{(j)}_k \gets 0$ else
	}
	// Now determine the price\;
	\If{$i \in \br^{(n)}$}{%
		// If player $i$ is part of the reference allocation he pays his valuation\;
		$p_i(z_i \mid \bx) \gets v_i$
	} 
	\Else{%
		// Otherwise he pays the critical value against the players in the reference allocation\;
		$p_i(z_i \mid \bx) \gets \inf \{v'_i: \; i \in \GRD(v'_i,\vals^{(n)}_{-i} \mid \allocs)\}$
	} 
} 
\Return{$p_i(z_i \mid \bx)$}\;
\caption{{\sc Pricing Rule Derived from $\GRD$} (Parametrized by $\ALG$)}
\label{alg:greedy-prices}
\end{algorithm}

\subsubsection{Proof of Theorem \ref{thm:greedy-perm}}

We prove the theorem in two steps. We first use permeability of the greedy allocation rule to establish Condition (a) (in Lemma \ref{lem:greedy-alpha}). We then show Condition (b). For this we first prove a novel combinatorial implication of permeability of the greedy allocation rule (in Lemma \ref{lem:zero-one}) by considering valuations that are either zero or one. We then use this property in a careful layering argument to establish Condition (b) (in Lemma \ref{lem:greedy-beta}).

\begin{lemma}\label{lem:greedy-alpha}
Let $\ALG$ be any allocation rule. Suppose that the greedy allocation rule $\GRD$ is $\gamma$-permeable for a subinstance-closed collection of binary, single-parameter problems $\Pi$. Then the pricing rule described in Algorithm \ref{alg:greedy-prices} fulfills Condition~(a) of Definition~\ref{def:ab2} with $\alpha = \gamma$ with respect to allocation rule $\ALG$ and the canonical exchange-feasible sets $(\feasx)_{\bx \in X}$.
\end{lemma}

\begin{proof}
Let $\allocs \in \feas$. We will show that $p_i(x_i \mid \allocs_{[i-1]}) \geq \frac{1}{\gamma} \cdot (\bv(\br^{(i-1)}) - \bv(\br^{(i)}))$. By a telescoping-sum argument, this then implies
\begin{align*}
\sum_{i \in N} p_i(x_i \mid \allocs_{[i-1]}) & \geq \sum_{i \in N} \frac{1}{\gamma} \cdot \bigg(\bv(\br^{(i-1)}) - \bv(\br^{(i)})\bigg) \\
&  = \frac{1}{\gamma} \cdot \bigg(\bv(\br^{(0)}) - \bv(\br^{(n)})\bigg) \geq \frac{1}{\gamma} \cdot \bigg(\bv(\ALG(\bv)) - \bv(\OPT(\bv, \feasx))\bigg),
\end{align*}
where the last step follows from the fact that $\br^{(0)} = \ALG(\bv)$ and $\br^{(n)} \in \feasx$.

So, it only remains to show $p_i(x_i \mid \allocs_{[i-1]}) \geq \frac{1}{\gamma} \cdot (\bv(\br^{(i-1)}) - \bv(\br^{(i)}))$. Observe that if $x_i = 0$, we have $\br^{(i-1)} = \br^{(i)}$ and this claim follows trivially. So, consider an arbitrary player $i$ for which $x_i = 1$. If $i \in \br^{(i-1)}$ then $\br^{(i-1)} \setminus \br^{(i)} = \{i\}$. So $p_i(1 \mid \allocs_{[i-1]}) = v_i$, while $\bv(\br^{(i-1)}) - \bv(\br^{(i)}) = v_i$ and the claim is true.

Otherwise, $i \not\in \br^{(i-1)}$, and we will first use smoothness with respect to subinstances to bound the size of the set
$\br^{(i-1)} \setminus \br^{(i)}$. For a fixed $\epsilon > 0$, define $\bv'$ by setting $v_i' = \tau_i^{\GRD}(\vals^{(i-1)}_{-i} \mid \allocs_{[i-1]}) + \epsilon$, $v_j' = v_j$ for $j \in \br^{(i-1)} \setminus \br^{(i)}$, and $v'_j = 0$ for all other $j$.

Now player $i \in \GRD(\bv' \mid \bx_{[i-1]} \cup \br^{(i)})$ by definition of $\bv'$ and $v'_i$ in particular, while for each player $j \in \br^{(i-1)} \setminus \br^{(i)}$ we have $j \not\in \GRD(\bv' \mid \allocs_{[i-1]} \cup \br^{(i)})$ because it cannot be added to $\allocs_{[i-1]} \cup \br^{(i)} \cup \{ i \} = \allocs_{[i]} \cup \br^{(i)}$ by definition of $\br^{(i)}$. Hence, the greedy critical values of each player $j \in \br^{(i-1)} \setminus \br^{(i)}$ must be at least $\tau_j^\GRD(\bv' \mid \allocs_{[i-1]} \cup \br^{(i)}) \geq v'_i$.

Since both player $i$ and the set of players $\br^{(i-1)} \setminus \br^{(i)}$ are feasible extensions to $\allocs_{[i-1]} \cup \br^{(i)}$,
we can use $\gamma$-permeability of the greedy allocation rule in the subinstance in which we hold $\allocs_{[i-1]} \cup \br^{(i)}$ fixed to obtain,
\[
|\br^{(i-1)}\setminus \br_{(i)}| \cdot v'_i \leq \sum_{j \in \br^{(i-1)}\setminus \br^{(i)}} \tau_j^{\GRD}(\bv' \mid \allocs_{[i-1]} \cup \br^{(i)})
\leq \gamma \cdot \bv'\left(\GRE(\bv' \mid \allocs_{[i-1]} \cup \br^{(i)})\right)
= \gamma \cdot v'_i.
\]
Cancelling $v'_i$ shows that $|\br^{(i-1)} \setminus \br^{(i)}| \leq \gamma$.

To show the claim it now suffices to observe that the greedy critical value of player $i$ in the subinstance
where we hold $\allocs_{[i-1]}$ fixed under the original valuations is the highest value of a player
$j \in \br^{(i-1)} \setminus \br^{(i)}$. Namely,
\[
	p_i(1 \mid \allocs_{[i-1]})
	= \max_{j \in \br^{(i-1)} \setminus \br^{(i)}} v_j
	\geq \frac{1}{\gamma} \cdot \sum_{j \in \br^{(i-1)} \setminus \br^{(i)}} v_j = \bv(\br^{(i-1)}) - \bv(\br^{(i)}),
\]
which concludes the proof.
\end{proof}

\begin{lemma}\label{lem:zero-one}
Suppose that the greedy allocation rule $\GRD$ is $\gamma$-permeable for a subinstance-closed collection of problems of binary, single-parameter problems $\Pi$. Consider any problem $\pi \in \Pi$ with feasibility structure $\feas$. Furthermore, let $B_0 \supseteq B_1 \supseteq \ldots \supseteq B_n$ and $A_0 \subseteq A_1 \subseteq \ldots \subseteq A_n$ with $B_t \cup A_t \in \feas$ for all $t$.
Consider a set $C$ that fulfills $C \cup A_n \in \F$ and for every $i \in C$ there is a $t \in \{0, 1, \ldots, n\}$ with $i \in B_t$ or $\{ i \} \cup B_t \cup A_t \not\in \F$. Then we have $\lvert C \rvert \leq \gamma \cdot \lvert B_0 \rvert$.
\end{lemma}
\begin{proof}
Set $B_{n+1} = \emptyset$ and define $C_{n+1} = C$. Furthermore, define $C_t$ for $0 \leq t \leq n$ recursively as a maximal subset of the players in $C_{t+1} \setminus B_t$ such that $C_t \cup A_t \cup B_t \in \F$.

We will show that for all $t \in \{0, 1, \ldots, n\}$,
\[
	|C_{t+1} \setminus C_t| \leq \gamma \cdot |B_t \setminus B_{t+1}|.
\]

Consider some fixed $t$ and define $D := B_t \cap C_{t+1}$.

Now a crucial observation is that the set $C_{t+1} \setminus (C_{t} \cup D)$ is feasible holding $E := B_{t+1} \cup A_{t} \cup C_{t} \cup D$
fixed. This is because
\[
\big( C_{t+1} \setminus (C_{t} \cup D) \big) \cup B_{t+1} \cup A_{t} \cup C_{t} \cup D
\; = \; B_{t+1} \cup A_{t} \cup C_{t+1}
\; \subseteq \; B_{t+1} \cup A_{t+1} \cup C_{t+1}
\; \in \; \F.
\]
Further note that by the way we have chosen $C_{t}$ we know that $B_{t}\setminus (B_{t+1} \cup D)$ is another feasible extension to $E$ because
\[
\left( B_{t} \setminus (B_{t+1} \cup D) \right) \cup B_{t+1} \cup A_{t} \cup C_{t} \cup D
\; = \; B_{t} \cup A_{t} \cup C_{t}
\; \in \; \F.
\]

To apply $\gamma$-permeability in the subinstance where we hold $E$ fixed, define a valuation profile $\bar{\bv}$ by setting $\bar{v}_i = 1$ for $i \in B_t \setminus (B_{t+1} \cup D)$ and $0$ otherwise. Now, for every $i \in C_{t+1} \setminus (C_t \cup D)$, we have
\[
\tau_i^{\GRE}(\bar{\bv} \mid E) = 1.
\]
This is due to the maximality of $C_t$: If for some $i \in C_{t+1} \setminus (C_t \cup D)$ this value is $0$, then also $\{ i \} \cup \left( B_t \setminus (B_{t+1} \cup D) \right) \cup B_{t+1} \cup A_t \cup C_t \cup D \in \F$.

So by permeability,
\begin{align*}
	\lvert C_{t+1} \setminus (C_{t} \cup D) \rvert = \sum_{i \in C_{t+1} \setminus (C_t \cup D)} \tau_i^{\GRE}(\bar{\bv} \mid E) \leq \gamma \cdot \bar{\bv}(\GRE(\bar{\bv} \mid E)) = \gamma \cdot \lvert B_{t} \setminus (B_{t+1} \cup D) \rvert,
\end{align*}
and therefore
\[
\lvert C_{t+1} \setminus C_{t} \rvert = \lvert D \rvert + \lvert C_{t+1} \setminus (C_{t} \cup D) \rvert \leq \lvert D \rvert + \gamma \cdot \lvert B_t \setminus (B_{t+1} \cup D) \rvert \leq \gamma \cdot \lvert B_t \setminus B_{t+1} \rvert.
\]

We now obtain the desired bound on the size of the set $C$ by summing the previous inequality over all $t$ and
using that it becomes a telescoping sum
\[
	|C| + |C_0| = |C_{n+1}| + |C_0|
	= \sum_{t = 0}^{n} |C_{t+1} \setminus C_t|
	\leq \gamma \cdot \sum_{t = 0}^{n}  |B_{t} \setminus B_{t+1}|
	= \gamma \cdot (|B_0| - |B_{n+1}|) = \gamma \cdot |B_0|.
\]
It remains to show that all players in $C$ will be covered (i.e., that $C_0 = \emptyset)$. This follows from the fact
that for each player $i \in C$ by the definition of $C$ there exists a $t$ such that $i \in B_t$ or $i$ does not fit into $B_t \cup A_t$
and, thus, in either case $i \not\in C_t$.
\end{proof}

\begin{figure}
\centering
\begin{tikzpicture}
\draw[very thick, gray]
  (-0.5,0) -- (8.5,0);
  \draw[dashed, thick, gray]
    (8,0.75+0.1)  -- (0,0.75+0.1) node[xshift=-0.5cm] {$v_{(1)}$};
  \draw[dashed, thick, gray]
    (8,1.25+0.1)  -- (0,1.25+0.1) node[xshift=-0.5cm] {$v_{(2)}$};
  \draw[dashed, thick, gray]
    (8,1.75+0.1)  -- (0,1.75+0.1) node[xshift=-0.5cm] {$v_{(3)}$};
  \draw[dashed, thick, gray]
    (8,2.75+0.1)  -- (0,2.75+0.1) node[xshift=-0.5cm] {$v_{(4)}$};
  \draw[dashed, thick, gray]
    (8,3.25+0.1)  -- (0,3.25+0.1) node[xshift=-0.5cm] {$v_{(5)}$};
\foreach \x/\y/\z in {0/2.75/dgray,1/1.75/lgray,2/0.75/white,3/1.25/lgray,4/3.25/dgray,5/0.75/white,6/2.75/dgray,7/0.75/white}
  \draw[rounded corners=5pt,thick,fill=\z]
    (\x+0.25,0.1) rectangle ++(0.5,\y);
\end{tikzpicture}
\caption{Our proof that the pricing rule derived from the greedy allocation rule satisfies Condition (\ref{condb}), relies on the fact that we can chop the valuation and price space into discrete layers, which reduces the problem to 0/1-valuations.} \label{fig:layering}
\end{figure}
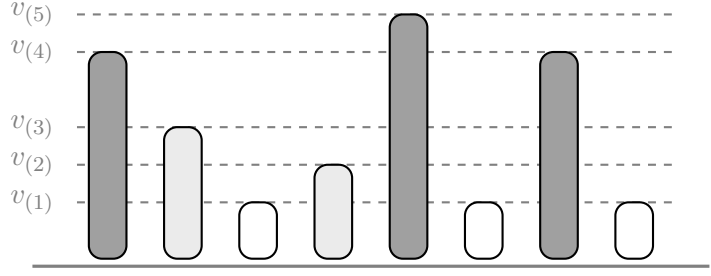

\begin{lemma}\label{lem:greedy-beta}
Let $ALG$ be any allocation rule. Suppose that the greedy allocation rule $\GRD$ is $\gamma$-permeable for a subinstance-closed collection of binary, single-parameter problems $\Pi$. Then the pricing rule described in Algorithm \ref{alg:greedy-prices} fulfills Condition~(b) of Definition~\ref{def:ab2} with $\beta_1 = 0$ and $\beta_2 = \gamma$ with respect to allocation rule $\ALG$ and the canonical exchange-feasible sets $(\feasx)_{\bx \in X}$.
\end{lemma}


\begin{proof}
Consider an arbitrary feasible $\allocs \in \F$ and an arbitrary feasible extension $\allocs' \in \feasx$. We want to show that
\[
	\sum_{i \in N} p_i(x_i' \mid \allocs_{[i-1]}) \leq \gamma \cdot \bv(\ALG(\bv)).
\]

We will prove this claim through a layering argument; as depicted in Figure \ref{fig:layering}. To this end, let $v_{(j)}$ be the $j$-th highest value of $v_1, \ldots, v_n$; furthermore $v_{(n+1)} = 0$. For each $j \in [n]$, let $S^j$ denote the set of players with value at least $v_{(j)}$ and let $T^j$ denote the set of players with $x'_i = 1$ that see a price $p_i(1 \mid \allocs_{[i-1]})$ of at least $v_{(j)}$.

We now apply Lemma~\ref{lem:zero-one} for each $j \in [n]$ by setting $A_t = \allocs_{[t]}$, $B_t = \br^{(t)} \cap S^j$, $C = T^j$. Note that $C \cup A_n \in \F$ and for every $i \in C$ there is a $t \in \{0, 1, \ldots, n\}$ such that $i \in B_t$ or $\{ i \} \cup B_t \cup A_t \not\in \F$. We obtain $\lvert T^j \rvert \leq \gamma \cdot \lvert \br^{(0)} \cap S^j \rvert$.

We conclude that
\begin{align*}
	\sum_{i \in N} p_i(x_i' \mid \allocs_{[i-1]}) & = \sum_{i \in N} \sum_{j=1}^{n} \ind_{i \in T^j} \cdot (v_{(j)} - v_{(j+1)}) \\
	&= \sum_{j=1}^{n} \lvert T^j \rvert \cdot (v_{(j)} - v_{(j+1)})\\
	&\leq \sum_{j=1}^{n} \gamma \cdot \lvert \br^{(0)} \cap S^j \rvert \cdot (v_{(j)} - v_{(j+1)})\\
	&= \gamma \cdot \bv(\ALG(\bv))\enspace,
\end{align*}
where the first equality holds by definition of the sets $T^j$, the second equality is basic calculus, the inequality follows from Lemma \ref{lem:zero-one} as argued above, and the final equality holds by definition of $\br^{(0)} = \ALG(\bv)$ and the sets $S^j$.
\end{proof}


\subsection{Proof of Theorem \ref{thm:opt}}
\label{app:opt}

In this subsection we prove Theorem \ref{thm:opt}, which claims that $(\lambda,\mu)$-smoothness of the pay-your-bid mechanism based on the welfare-maximizing allocation rule for subinstance-closed collection of binary, single-parameter problems $\Pi$ implies the existence of a weakly $(1,(\mu+1)/\lambda)$-balanced pricing rule. By Theorem \ref{thm:permeable} it suffices to show the following theorem.

\begin{theorem}\label{thm:opt-perm}
Let $\ALG$ be any allocation rule. Suppose that the welfare-maximizing allocation rule $\OPT$ is $\gamma$-permeable for a subinstance-closed collection of binary, single-parameter feasibility problems $\Pi$. Then there exists a pricing rule that is weakly $(1,0,\gamma^2)$-balanced with respect to $\ALG$ and the canonical exchange-feasible sets $(\feas_\bx)_{\bx \in X}$.
\end{theorem}


As in the case of greedy we first describe the construction of the prices, and then we show that these prices are balanced. 


\subsubsection{Construction of the Prices}
\label{app:opt-prices}

We define the price $p_i(z_i \mid \allocs)$ for player $i \in N$ and outcome $z_i \in \{0,1\}$ and arbitrary but fixed valuation profile $\vals$ and allocation $\allocs \in \feas$ through Algorithm \ref{alg:opt-prices}. 
For this section, define $\OPT(\bv \mid \bx)$ as the allocation that results by padding the welfare-maximizing allocation for valuation profile $\bv$ over $\contraction$ with empty allocations. 

As in the the case of the greedy allocation rule, we again set $p_i(0 \mid \bx) = 0$ and we compute $p_i(1 \mid \bx)$ via reference allocations and reference valuations. We again define the initial reference allocation as $\br^{(0)} = \ALG(\vals)$ and the initial reference valuations by setting $v^{(0)}_j = v_j$ for $j \in \br^{(0)}$ and $v^{(0)}_j = 0$ otherwise.
The subsequent reference allocations and valuations are defined recursively as $r^{(i)} = \OPT(v^{(i-1)} \mid \allocs_{[i]})$ and $v^{(i)}_j = v^{(i-1)}_j = v_j$ for $j \in \br^{(0)}$ and $v^{(i)}_j = 0$ otherwise. We then set $p_i(1 \mid \bx) = v_i$ if $i \in \br^{(n)}$ and $p_i(1 \mid \bx) = \inf\{v'_i\mid i \in \OPT(v'_i,\valsmi^{(n)} \mid \allocs)\}$ otherwise. Note that this definition immediately implies that $p_i(x_i \mid \allocs_{[i-1]}) = \bv(\br^{(i-1)}) - \bv(\br^{(i)})$.



By substituting all occurrences of $n$ with $i-1$ we obtain the formula for the price $p_i(1 \mid \allocs_{[i-1]})$. We use the shorthand $\tau^{\OPT}(\valsmi^{(i-1)} \mid \allocs_{[i-1]}) := \inf\{v'_i\mid i \in \OPT(v'_i,\valsmi^{(i-1)}) \mid \allocs_{[i-1]})\}$.

\begin{algorithm}
\DontPrintSemicolon
\KwIn{$z_i \in \{0,1\}$, $\bv$, $\bx \in \feas$}
\KwOut{$p_i(z_i \mid \allocs)$}
\If{$z_i = 0$}{%
// In this case the price is simply zero\;
$p_i(z_i \mid \allocs) = 0$%
} 
\Else{%
	// First determine reference allocation and valuations\;
	$\br^{(0)} \gets \ALG(\bv)$,
	$v^{(0)}_k \gets v_k$ if $k \in \br^{(0)}$ and $v^{(0)}_k \gets 0$ otherwise\;
	\For{$j \gets 1$ {\bf to} $n$}{
		$\br^{(j)} \gets \OPT(\vals^{(j-1)} \mid \bx_{[j]})$,
		$v^{(j)}_k \gets v_k^{(j-1)} = v_k$ if $k \in \br^{(j)}$ and
		$v^{(j)}_k \gets 0$ else
	}
	// Now determine the price\;
	\If{$i \in \br^{(n)}$}{%
		// If player $i$ is part of the reference allocation he pays his valuation\;
		$p_i(z_i \mid \bx) \gets v_i$
	} 
	\Else{%
		// Otherwise he pays the critical value against the players in the reference allocation\;
		$p_i(z_i \mid \bx) \gets \inf \{v'_i: \; i \in \OPT(v'_i,\vals^{(n)}_{-i} \mid \allocs)\}$
	} 
} 
\Return{$p_i(z_i \mid \bx)$}\;
\caption{{\sc Pricing Rule Derived from $\OPT$} (Parametrized by $\ALG$)}
\label{alg:opt-prices}
\end{algorithm}


\subsubsection{Proof of Theorem \ref{thm:opt-perm}}

We again proceed in two steps. We first show Condition (a) (in Lemma \ref{lem:opt-alpha} below). Afterwards we show that the permeability of $\OPT$ provides an upper bound on the permeability of $\GRD$ and that the critical prices with respect to $\OPT$ are not much higher than those with respect to $\GRD$ (in Lemmas \ref{lem:greedyoptpermeability} and \ref{lem:greedyoptthresholds}). This allows us to bound Condition (b) using the same machinery that we used in the previous section (in Lemma \ref{lem:opt-beta})

\begin{lemma}\label{lem:opt-alpha}
Let $\ALG$ be any allocation rule. Suppose that the welfare-maximizing allocation rule $\OPT$ is $\gamma$-permeable for a subinstance-closed collection of problems of binary, single-parameter problems $\Pi$. Then the pricing rule described in Algorithm \ref{alg:opt-prices} fulfills Condition~(a) of Definition~\ref{def:ab2} with $\alpha = 1$ with respect to allocation rule $\ALG$ and the canonical exchange-feasible sets $(\feasx)_{\bx \in X}$.
\end{lemma}

\begin{proof}
Consider $\allocs \in \feas$. We defined prices exactly so that $p_i(x_i \mid \allocs_{[i-1]}) = \bv(\br^{(i-1)}) - \bv(\br^{(i)})$.  Therefore, using a telescoping-sum argument, we get
\[
\sum_{i \in N} p_i(x_i \mid \allocs_{[i-1]}) = \bv(\br^{(0)}) - \bv(\br^{(n)}).
\]
The claim now follows from the fact that $\br^{(0)} = \ALG(\bv)$ and $\br^{(n)} \in \feasx$.
\end{proof}

\begin{lemma}
\label{lem:greedyoptpermeability}
If the greedy allocation rule $\GRD$ is $\gamma^\GRE$-permeable and the welfare-maximizing allocation rule $\OPT$ is $\gamma^\OPT$-permeable for a subinstance-closed collection of of binary, single-parameter problems $\Pi$, then $\gamma^\OPT \geq \gamma^\GRE$.
\end{lemma}

\begin{proof}
We only have to show that for all $\allocs \in \feas$, $\allocs' \in \feasx$ and all $\bv$, we have
\[
	\sum_{i \in \allocs'} \tau_i^\GRE(\valsmi \mid \allocs) \leq \gamma^\OPT \cdot \bv(\GRE(\bv \mid \allocs)).
\]

Let $\by = \allocs' \cap \GRE(\bv \mid \allocs)$. Observe that because $\by \subseteq \GRE(\bv \mid \allocs)$, we have $\GRE(\bv \mid \allocs \cup \by) = \GRE(\bv \mid \allocs)$.
Define a valuation profile $\bv'$ by setting $v_i' = v_i$ for all $i \in \allocs' \setminus Q$ and $v_i' = 0$ otherwise. By loser independence of greedy, $\GRE(\bv \mid \allocs \cup \by) = \GRE(\bv' \mid \allocs \cup \by)$.  Furthermore, $\OPT(\bv' \mid \allocs \cup \by) = \GRE(\bv' \mid \allocs \cup \by)$.

We claim that for $i \in \bx'\setminus \by$ we have
\[
\tau_i^\GRE(\valsmi \mid \allocs) \leq \tau_i^\GRE(\valsmi' \mid \allocs \cup \by) \leq \tau_i^\OPT(\valsmi' \mid \allocs \cup \by) \enspace.
\]

For the first inequality we use that $\tau_i^\GRE(\valsmi \mid \allocs)$ is the value $v_j$ of some $j$ that has to be outbid by player $i$. By fixing another set $\by$, this value can only go up because the options are limited further. Also, when fixing $\allocs \cup \by$, replacing the valuations by $\bv'$ has no influence because the set of players that are selected remains unchanged.

The second inequality holds because under $\bv'$ player $i$ in order to win when we hold $\bx \cup \by$ fixed has to force some subset $\bz \subseteq \GRE(\valsmi' \mid \allocs \cup \by) = \OPT(\valsmi' \mid \allocs \cup \by)$ out of the solution. Under $\OPT$ his payment is the sum of the respective players' valuations, under $\GRD$ it is just the highest valuation of any such player.

%


On the other hand, for players $i \in \by$, because $\by \subseteq \GRE(\vals \mid \allocs)$, the greedy critical value $\tau_i^\GRE(\valsmi \mid \allocs)$ is at most $v_i$.

Using these two bounds on $\tau_i^\GRE(\valsmi \mid \allocs)$ we obtain,
\begin{align*}
	\sum_{i \in \allocs'} \tau_i^\GRE(\valsmi \mid \allocs) & \leq \sum_{i \in \by} v_i + \sum_{i \in \allocs' \setminus \by} \tau_i^\GRE(\valsmi \mid \allocs) \\
	& \leq \sum_{i \in \by} v_i + \sum_{i \in \allocs' \setminus \by} \tau_i^\OPT(\valsmi' \mid \allocs \cup \by) \\
	& \leq \sum_{i \in \by} v_i + \gamma^\OPT \cdot \bv'(\OPT(\bv' \mid \allocs \cup \by))\\
	& = \sum_{i \in \by} v_i + \gamma^\OPT \cdot \bv'(\GRE(\bv' \mid \allocs \cup \by))\\
	& \leq \gamma^\OPT \cdot \sum_{i \in \by} v_i + \gamma^\OPT \cdot \bv'(\GRE(\bv' \mid \allocs \cup \by)) \\
	& = \gamma^\OPT \cdot \bv(\OPT(\bv \mid \allocs)\enspace,
\end{align*}
where the third inequality use $\gamma^\OPT$-permeability of the welfare-maximizing allocation rule $\OPT$ in the subinstance in which we hold $\bx \cup \by$ fixed, the subsequent equality holds by the definition of $\bv'$, the fourth inequality uses that $\gamma^\OPT \geq 1$, and the final equality holds by the definition of $\by$ and $\bv'$.
\end{proof}

\begin{lemma}
\label{lem:greedyoptthresholds}
If the greedy allocation rule $\GRD$ is $\gamma^\GRD$-permeable for a subinstance-closed collection of of binary, single-parameter problems $\Pi$, then
\[
	\tau_i^\OPT(\bv^{(i-1)} \mid \bx_{[i-1]}) \leq \gamma^\GRE \cdot \tau_i^\GRE(\bv^{(i-1)} \mid \bx_{[i-1]})
\]
\end{lemma}
\begin{proof}
Note that for valuations $\vals^{(i-1)}$ both $\GRD$ and $\OPT$ over $\feas/\bx$ return the same set of players. 
The same is true if we drop any player $j$ from $\vals^{(i-1)}$. Dropping player $i$ we can define $\by = \GRD(\valsmi^{(i-1)} \mid \bx_{[i-1]}) = \OPT(\valsmi^{(i-1)} \mid \bx_{[i-1]})$.

Now consider player $i$ bidding $b'_i = \tau_i^\GRD(\valsmi^{(i-1)} \mid \bx_{[i-1]}) + \epsilon$. Then $i \in \GRD(b'_i,\valsmi^{(i-1)} \mid \bx_{[i-1]})$. The addition of player $i$ causes the removal of a (possibly empty) subset of players $\bz \subseteq \by$. That is, $\GRD(b'_i,\valsmi^{(i-1)} \mid \bx_{[i-1]}) = (\by \setminus \bz) \cup \{i\}$.

Define valuations $\vals'$ by setting $v'_i = b'_i$, $v'_j = v^{(i-1)}_j$ for $j \in \bz$, and $v'_j = 0$ for every other player $j$. Consider the subinstance in which we hold $\by \setminus \bz$ fixed. Since both player $i$ and the set of players $\bz$ are feasible extensions we can apply $\gamma^\GRD$-permeability of the greedy allocation rule to obtain
\[
	|\bz| \cdot b'_i
	= \sum_{j \in \bz} \tau_j^{\GRD}(\vals' \mid \by \setminus \bz)\\
	\leq \gamma^\GRE \cdot \bv'(\GRD(\bv' \mid \by \setminus \bz))
	= \gamma^\GRE \cdot b'_i,
\]
and therefore $|\bz| \leq \gamma^\GRD.$

The final step is now to observe that the critical value $\tau_i^\OPT(\bv^{(i-1)} \mid \allocs_{[i-1]})$ of player $i$ under the
welfare-maximizing allocation rule is at most $|T|$ times the critical value $p_i^\GRE(\bv^{(i-1)} \mid \allocs_{[i-1]})$
of player $i$ under the greedy allocation rule.
To see this let $\bz' \subseteq \by$ be the set with the smallest value such that
$(\by \setminus \bz') \cup \{i\} \in \F$. Then,
\[
	\tau_i^\OPT(\bv^{(i-1)} \mid \allocs_{[i-1]})
	= \sum_{j \in \bz'} v_j
	\leq \sum_{j \in \bz} v_j \leq |T| \cdot \max_{j \in \bz} v_j
	= |\bz| \cdot \tau_i^\GRE(\bv^{(i-1)} \mid \allocs_{[i-1]}),
\]
where we used that $\bz$ is some subset of $\by$ such that $(\by \setminus \bz) \cup \{i\} \in \F$
and so its combined value can only be larger than that of $\bz'$.
\end{proof}

\begin{lemma}\label{lem:opt-beta}
Let $\ALG$ be any allocation rule. Suppose that the welfare-maximizing allocation rule $\OPT$ is $\gamma$-permeable for a subinstance-closed collection of binary, single-parameter problems $\Pi$. Then the pricing rule described in Algorithm \ref{alg:opt-prices} fulfills Condition~(b) of Definition~\ref{def:ab2} with $\beta_1 = 0$ and $\beta_2 = \gamma^2$ with respect to allocation rule $\ALG$ and the canonical exchange-feasible sets $(\feasx)_{\bx \in X}$.
\end{lemma}



\begin{proof}
Consider an arbitrary allocation $\allocs \in \feas$ and feasible extension $\allocs' \in \feasx$. We want to show that
\[
	\sum_{i \in N} p_i(x'_i \mid \allocs_{[i-1]}) \leq (\gamma^{\OPT})^2 \cdot \vals(\ALG(\vals)).
\]

We will again show this claim through layering. This time, however, we need the layering to arbitrarily fine-grained. We will specify the granularity by $\epsilon > 0$. For each $j \in \mathbb{N}$, let $S^j$ denote the set of players with value at least $j \cdot \epsilon$. Let $T^j$ denote the set of players with $x'_i = 1$ that see a price $p_i^\OPT(x'_i \mid \allocs_{[i-1]})$ of at least $\gamma^{\GRD} \cdot j \cdot \epsilon$.

For any fixed $j \in \mathbb{N}$ we now bound $\lvert T^j \rvert$ using Lemma~\ref{lem:zero-one}. We set $A_t = x_{[t]}$, $B_t = \br^{(t)} \cap S^j$, $C = T^j$. Note that $C \cup A_n \in \F$. We claim that for every $i \in C$ we have $i \in B_t$ or $\{ i \} \cup B_t \cup A_t \not\in \F$ for $t = i-1$. To see this, consider the two options how $p_i^\OPT(x'_i \mid \allocs_{[i-1]})$ can be set. If $i \in \br^{(i-1)}$, then $p_i^\OPT(x'_i \mid \allocs_{[i-1]}) = v_i$. As by definition $p_i^\OPT(x'_i \mid \allocs_{[i-1]}) \geq \gamma^\GRE \cdot j \cdot \epsilon \geq j \cdot \epsilon$, this implies $v_i \geq j \cdot \epsilon$ and so $i \in B_{i-1} = \br^{(i-1)} \cap S^j$. Otherwise, if $i \not\in \br^{(i-1)}$, then $p_i^\OPT(x'_i \mid \allocs_{[i-1]}) = \tau_i^\OPT(\valsmi^{(i-1)} \mid \allocs_{[i-1]})$. In this case, we can apply Lemma~\ref{lem:greedyoptthresholds} to get
\[
p_i^\OPT(x'_i \mid \allocs_{[i-1]}) = \tau_i^\OPT(\valsmi^{(i-1)} \mid \allocs_{[i-1]}) \leq \gamma^\GRE \cdot \tau_i^\GRD(\valsmi^{(i-1)} \mid \allocs_{[i-1]}).
\]
So, we know that $\tau_i^\GRD(\valsmi^{(i-1)} \mid \allocs_{[i-1]}) \geq j \cdot \epsilon$ and therefore $\{ i \} \cup \allocs_{[i-1]} \cup (\br^{(i-1)} \cap Y^j) \not\in \F$.

So, by Lemma~\ref{lem:zero-one}, we get $\lvert T^j \rvert \leq \gamma^\GRE \cdot \lvert \br^{(0)} \cap Y^j \rvert$.

We conclude that
\begin{align*}
	\sum_{i \in \allocs'} p_i^\OPT(x'_i \mid \allocs_{[i-1]})
	&\leq n \epsilon + \sum_{i \in \allocs'} \sum_{j=1}^{\infty} \ind_{i \in T^j} \cdot \gamma^\GRD \cdot \epsilon \\
	&= n \epsilon + \gamma^\GRD \cdot \sum_{j=1}^{\infty} |T^j| \cdot \epsilon \\
	&\leq n \epsilon + \gamma^\GRD \cdot \sum_{j=1}^{\infty} \gamma^\GRD \cdot \lvert \br^{(0)} \cap Y^j \rvert \cdot \epsilon\\
	&\leq (\gamma^\GRD)^2 \cdot \bv(\ALG(\bv)) + 2 n \epsilon.
\end{align*}
As this argument holds for any $\epsilon > 0$, we also have $\sum_{i \in \allocs'} p_i^\OPT(x'_i \mid \allocs_{[i-1]}) \leq (\gamma^\GRE)^2 \cdot \bv(\ALG(\bv))$. By Lemma \ref{lem:greedyoptpermeability}, $\gamma^\GRD \leq \gamma^\OPT$ and the claim follows.
\end{proof}

\begin{remark}
\label{rem:monotone}
When the prices defined by Algorithm \ref{alg:opt-prices} are non-decreasing then $\gamma$-permeability of $\OPT$ immediately implies that $\tau_i^\GRD(\bv^{(i-1)}_{-i} \mid \bx_{[i-1]}) \leq \tau_i^\GRD(\bv^{(i-1)}_{-i} \mid \bx) \leq \gamma \cdot \bv(\OPT(\bv, \feasx))$ implying that the pricing rule is $(\alpha,\beta)$-balanced with $\beta = \gamma$. In this case Theorem \ref{thm:opt-perm} can be strengthened to show the existence of a $(1,\gamma)$-balanced pricing rule.
\end{remark}

\end{document}